\documentclass{article}[10pt]
\usepackage{epsfig,pict2e}
\usepackage{amsmath}
\usepackage{amsthm}
\usepackage{amsfonts}
\usepackage{amssymb}
\usepackage{mathtools}
\usepackage{relsize}
\usepackage{tikz}
\usepackage{tikz-cd}
\usetikzlibrary{calc} 
\usepackage{stmaryrd}
\usepackage{geometry}
\usepackage[toc,page]{appendix}
\usepackage{graphicx}
\usepackage{caption}
\usepackage{subcaption}
\usepackage{hyperref}
\usepackage{array} 
\usepackage{paralist} 
\usepackage{verbatim} 
\newtheorem{theorem}{Theorem}[section]
\newtheorem*{definition*}{Definition}
\newtheorem{definition}[theorem]{Definition}
\newtheorem*{theorem*}{Theorem}
\newtheorem{proposition}[theorem]{Proposition}
\newtheorem*{proposition*}{Proposition}

\newtheorem*{lemma*}{Lemma}
\newtheorem*{claim*}{Claim}

\newtheorem*{corollary*}{Corollary}

\newtheorem*{example*}{Example}
\newtheorem*{conjecture*}{Conjecture}
\newtheorem{remark}[theorem]{Remark}
\newtheorem*{remark*}{Remark}

\newcommand{\Z}{\mathbb{Z}}

\newcommand{\C}{\mathbb{C}}

\newcommand{\p}{\mathbb{P}}

\renewcommand{\L}{\mathcal{L}}

\renewcommand{\H}{\mathcal{H}}
\newcommand{\E}{\mathcal{E}}
\newcommand{\Pic}{\text{Pic}}
\renewcommand{\div}{\text{div}}
\newcommand{\Div}{\text{Div}}

\newcommand{\pain}[1]{\text{P}_{\mathrm{#1}}}

\newenvironment{indented}{\begin{indented}}{\end{indented}}
\def\indented{\list{}{\itemsep=0\p@\labelsep=0\p@\itemindent=0\p@
   \labelwidth=0\p@\leftmargin=\mathindent\topsep=0\p@\partopsep=0\p@
   \parsep=0\p@\listparindent=15\p@}\footnotesize\rm}

\begin{document}
\begin{center}
\LARGE{\textbf{Full-parameter discrete Painlev\'e systems from non-translational Cremona isometries}}
\end{center}
\vspace{1ex}
\begin{center}
\normalsize{Alexander Stokes}\\
\noindent \small{Department of Mathematics, University College London,\\
Gower Street, London, WC1E 6BT, UK}
\end{center}
\normalsize


\begin{abstract} 
Since the classification of discrete Painlev\'e equations in terms of rational surfaces, there has been much interest in the range of integrable equations arising from each of the 22 surface types in Sakai's list. For all but the most degenerate type in the list, the surfaces come in families which admit affine Weyl groups of symmetries. Translation elements of this symmetry group define discrete Painlev\'e equations with the same number of parameters as their family of surfaces. While non-translation elements of the symmetry group have been observed to correspond to discrete systems of Painlev\'e-type through projective reduction, these have fewer than the maximal number of free parameters corresponding to their surface type. We show that difference equations with the full number of free parameters can be constructed from non-translation elements of infinite order in the symmetry group, constructing several examples and demonstrating their integrability. This is prompted by the study of a previously proposed discrete Painlev\'e equation related to a special class of discrete analogues of surfaces of constant negative Gaussian curvature, which we generalise to a full-parameter integrable difference equation, given by the Cremona action of a non-translation element of the extended affine Weyl group $\widetilde{W}(D_4^{(1)})$ on a family of generic $D_4^{(1)}$- surfaces.
\end{abstract}

\section{Introduction}

In a seminal paper, H. Sakai proposed a geometric framework for the study of discrete Painlev\'e equations \cite{SAKAI2001}. Sakai defined and classified families of surfaces generalising Okamoto's space of initial values for the continuous Painlev\'e equations \cite{OKAMOTO1979}, with discrete Painlev\'e equations defined in terms of the translation part of the affine Weyl symmetry group of each family. In addition to a classification scheme, the geometric approach offers a suite of tools for describing many apects of Painlev\'e equations, as demonstrated thoroughly in an important recent survey by Kajiwara, Noumi and Yamada \cite{KNY2017}. Further, recent work (e.g. \cite{DT2018}) has also pointed out the value of the geometric approach in identifying when a given system is equivalent to a known example, in which case one can make use of known special solutions, B\"acklund transformations, etc.  \\

It is now widely appreciated that Sakai's scheme classifies surfaces into a finite number of types, but does not further classify the equations belonging to each. Recently, much research has explored the range of inequivalent discrete systems which may be associated with each of the 22 surface types in Sakai's list. In particular, systems proposed to be of discrete Painlev\'e type independently of Sakai's scheme being placed within the geometric framework has shed much light on the infinite number of discrete Painlev\'e systems associated with each surface type. Examples include equations sharing the same space of initial values but corresponding to non-conjugate translations in the affine Weyl symmetry group \cite{JN2017, RG2009}, as well as elements of infinite order, though not translations, giving rise to difference equations with finer time evolution and less free parameters through projective reduction \cite{AHJN2016, HHNS2015, KNT2011, TAKENAWA2003}.\\

Our approach continues along these lines, in that we begin with a difference equation previously proposed as a discrete analogue of a special case of the third Painlev\'e equation, and place it within the geometric framework with a view to better understanding the features of Sakai's scheme. \\

The equation we consider was constructed in \cite{HOFFMAN1999}, through a discrete version of the process by which a special class of K-surfaces (surfaces of constant negative Gaussian curvature) are controlled by a particular case of the third Painlev\'e equation. The discrete analogues of this class are known as discrete Amsler surfaces, and are related to solutions of the discrete Sine-Gordon equation
\begin{equation} \label{SG}
Q_{l+1,m+1}Q_{l,m}= F(Q_{l+1,m}) F(Q_{l,m+1}),
\end{equation}
where $F(x)= \frac{1-k x}{k-x}$, with free complex parameter $k \not\in \left\{0, \pm1 \right\}$. Solutions corresponding to discrete Amsler surfaces satisfy an additional condition, which may be interpreted as invariance under Lorentz rotations of the frame (see \cite{HOFFMAN1999} for details). The reduction in the continuous setting involves considering solutions of the Sine-Gordon equation along the diagonal. In the discrete case, equation \eqref{SG}, together with the additional constraint, gives a system which may be iterated along a zigzag path, leading to the following system of ordinary difference equations:
\begin{subequations} \label{Q}
\begin{align}
Q_{2n+1} &= \frac{ \frac{1}{F(Q_{2n})}- \frac{Q_{2n-1}}{2n+1}}{Q_{2n-1} F(Q_{2n}) - \frac{1}{2n+1}}, \\
Q_{2n+2} &= \frac{1}{ Q_{2n} (F(Q_{2n+1}))^2}.
\end{align}
\end{subequations}
We rewrite this in terms of the variables $(f_n, g_n) = (Q_{2n}, Q_{2n-1})$, which gives  
\begin{subequations} \label{hoffmans}
\begin{align}
\bar{f} &= \frac{(k-\bar{g})^2}{f (k \bar{g}-1)^2}, \\
\bar{g} &= \frac{(k-f) \left( k fg - (2n+1) f - g +k(2n+1) \right)}{(kf-1) \left( k(2n+1) f g- f - (2n+1) g + k \right)},
\end{align}
\end{subequations}
where $(f,g) = (f_n,g_n)$ and $(\bar{f},\bar{g}) = (f_{n+1}, g_{n+1})$. \\

In what follows, we construct the space of initial values for the system \eqref{hoffmans}, which is a family of $D_4^{(1)}$-surfaces with one free parameter, rather than five in the generic case. The Cremona isometry corresponding to iteration of the system is a `twisted translation': the composition of a translation in the weight lattice and a Dynkin diagram automorphism which preserves the translation direction. The parameter specialisation causes the action on parameter space to coincide with that of the untwisted translation, so the system is in this sense a projective reduction, whose properties we compare to previously studied examples \cite{KNT2011,HHNS2015,AHJN2016,TAKENAWA2003,JN2017}. This prompts us to construct a generic version (in the sense that it has the maximal number of free parameters for its surface type), for which the action on parameter space is translational except for a permutation of parameters. We express this explicitly as a system of difference equations, and show that it is integrable in the sense of vanishing algebraic entropy \cite{BV1999}. We then show how full-parameter difference equations can be obtained in this way from the Cremona action of elements of infinite order on a family of generic surfaces, though they will not necessarily be of purely additive, multiplicative or elliptic type.

\subsection{Background}

The continuous Painlev\'e equations $\pain{I}$-$\pain{VI}$ are six nonlinear second-order ordinary differential equations, the study of which has become one of the cornerstones of the field of integrable systems. Beginning in the 1990's, important steps were made towards defining and understanding their discrete analogues, through the proposal by Ramani and Grammaticos, together with Papageorgiou, of \emph{singularity confinement} \cite{RPG1991} as the discrete counterpart to the Painlev\'e property. This led to the technique of deautonomisation by singularity confinement, which when applied to the Quispel-Roberts-Thompson (QRT) mappings \cite{QRT1989, QRT1988} resulted in the discovery of many integrable equations of discrete Painlev\'e type (see e.g. \cite{RGH1991}). \\

The property of the continuous Painlev\'e equations that would prove most useful in formulating the theory in the discrete setting was a geometric one. Okamoto demonstrated that each of the Painlev\'e equations could be associated with a family of rational surfaces obtained from $\C\p^2$ through a sequence of nine blowups \cite{OKAMOTO1979}. This was through the construction of a bundle over the independent variable space (excluding locations of singularities of the equation) whose fibres are rational surfaces with certain curves removed. These curves are known as \emph{inaccessible divisors}, and form the subset of each fibre on which the Painlev\'e vector field diverges. On the resulting bundle, known as \emph{Okamoto's space of initial conditions}, the Painlev\'e equation is regularised in the sense that the Cauchy problem for the differential equation is well-posed at every point. In particular, the intersection configuration of the removed curves was observed to be given by a Dynkin diagram of affine type, complementary to the one associated with the affine Weyl group of B\"acklund transformation symmetries of the equation. \\

It was shown that Okamoto's space essentially determines the Painlev\'e equation in each case \cite{MMT1999,MATSUMIYA1997, TS1997}. The curves removed from each fibre were observed to give a decomposition into irreducible components of a representative of the anticanonical divisor class of the surface, which led to the idea of classifying rational surfaces with such a configuration of curves via the notion of an \emph{Okamoto-Painlev\'e pair} \cite{ST2002, STT2002}. This is a pair $(S, D)$ of a smooth rational surface $S$ and a representative $D$ of its anticanonical divisor class which is of \emph{canonical type}. That is, its decomposition into irreducible components $D = \sum_{i=1}^r m_i D_i$ is such that 
\begin{equation}
\text{deg} \left.D \right|_{D_i} = D \cdot D_i = 0, \quad \quad \text{for all }i = 1, \dots r.
\end{equation}
Sakai used this idea to great effect in the formulation of discrete Painlev\'e equations, defining \emph{generalised Halphen surfaces} of index zero as rational surfaces which have a unique representative of the anticanonical class of canonical type \cite{SAKAI2001}. Families of such surfaces admit affine Weyl group symmetries, from which discrete Painlev\'e equations can be constructed in an explicit way. Further, the surfaces are classified into 22 types according to the intersection configuration and homology of the irreducible components of their anticanonical divisors, which is now regarded as the definitive classification scheme for second-order discrete systems of Painlev\'e type.\\

In the years since the publication of Sakai's paper \cite{SAKAI2001}, many known discrete analogues of the Painlev\'e equations have been found to fit naturally into the geometric framework. Studying examples in this way has led to a better understanding of the range of equations associated with each surface type on the list. For example, study of the geometry of an elliptic difference equation \cite{RCG2009,AHJN2016,JN2017}, obtained as a reduction of the discrete Krichever-Novikov equation \cite{ADLER1998, AS2004} (Q4 in the Adler-Bobenko-Suris classification \cite{ABS2003}), led to the construction of a new elliptic discrete Painlev\'e equation, demonstrating the non-uniqueness of the example presented in \cite{MSY2003,SAKAI2001} associated with the elliptic surface type at the top of Sakai's list. Further, while in Sakai's theory discrete Painlev\'e equations are defined to correspond to elements of the translation part of the symmetry group of a family of surfaces, the degeneration from a $q$-discrete analogue of $\pain{III}$ \cite{RGKT2000} to a $q$-discrete $\pain{II}$ \cite{RG1996} was formulated in the geometric setting as the process of \emph{projective reduction}, in which discrete systems arise from elements of infinite order which are not  translations in the affine Weyl group sense. \\

Here we will briefly recount this important example, following \cite{KNT2011}, as it will illustrate concepts relevant to the interpretation of our results. Consider the extended affine Weyl group 
\begin{equation}
\widetilde{W}((A_2 + A_1)^{(1)}) = \left< s_0, s_1, s_2 , w_0 , w_1 , \sigma \right>,
\end{equation}
defined by the fundamental relations
\begin{subequations}
\begin{align}
s_i^2 = (s_i &s_{i+1})^3 = 1, \quad w_j^2 = (w_j w_{j+1})^{\infty} = 1, \quad (s_i w_j)^2 = 1, \\
&\sigma^6 = 1, \quad s_i \sigma = \sigma s_{i+1}, \quad w_j \sigma = \sigma w_{j+1},
\end{align}
\end{subequations}
where $i \in \Z / 3 \Z$, $j \in \Z / 2 \Z$, and $(w_j w_{j+1})^{\infty} = 1$ means that the element $w_j w_{j+1}$ is of infinite order. In \cite{KNT2011}, Kajiwara, Nakazono and Tsuda considered a left-action of this group by birational mappings on the field of rational functions of parameters $a_0, a_1, a_2, c$, and variables $f_0, f_1, f_2$ subject to the constraint
\begin{equation}
f_0 f_1 f_2 = a_0 a_1 a_2 c^2 .
\end{equation}
Letting $q = a_0 a_1 a_2$ and $\pi = \sigma^2$, the action of the element $T_1 = \pi s_2 s_1$ on the parameters is given by
\begin{equation}
T_1.  \left( a_0, a_1, a_2 , c \right) =  \left( q a_0, q^{-1} a_1, a_2 , c \right).
\end{equation}
The element $T_1$ is a translation by a weight of the root system $(A_2+A_1)$, and its action on the variables leads to a system of first-order $q$-difference equations. To be precise, this is given by repeated iteration of the mapping
\begin{equation}
T_1 : \left( \begin{array}{c} a_0, a_1, a_2 \\ c  \end{array} ; f_0 , f_1, f_2 \right) \mapsto \left( \begin{array}{c} q a_0, q^{-1} a_1, a_2 \\ c  \end{array} ; T(f_0) , T(f_1), T(f_2) \right),
\end{equation}
where by letting $F_n = T_1^n(f_0)$ and $G_n = T_1^n(f_1)$, we obtain
\begin{equation} \label{qP3}
G_{n+1} G_n = \frac{q c^2}{F_n} \frac{1 + q^n a_0 F_n}{q^n a_0 + F_n}, \quad \quad F_{n+1} F_{n} = \frac{q c^2}{G_{n+1}} \frac{1 + q^n a_0 a_2 G_{n+1}}{q^n a_0 a_2 + G_{n+1}},
\end{equation}
which is a previously known $q$-discrete analogue of $\pain{III}$ \cite{RGKT2000}. We note that this equation fits Sakai's definition of a discrete Painlev\'e equation associated with a family of generic $A_5^{(1)}$-surfaces. \\

A related equation, obtained as a $q$-discrete analogue of $\pain{II}$ \cite{RG1996}, is given by
\begin{equation} \label{qP2}
F_{n+1} F_{n-1} = \frac{q c^2}{F_n} \frac{1 + q^{n/2} a_0 F_n}{q^{n/2} a_0 + F_n},
\end{equation}
where $a_0, c$ are free parameters, with $q$ corresponding to the multiplicative timestep in the independent variable. The process by which this is obtained from equation \eqref{qP3} was referred to as \emph{symmetrisation} of discrete Painlev\'e equations \cite{RGH1991} after the corresponding process on the level of QRT maps, and was found to correspond in the geometric setting to structural features of the affine Weyl group and its birational representation \cite{KNT2011}, in a way which we now recall. \\

With the same birational representation of $\widetilde{W}((A_2+A_1)^{(1)})$, consider the element $R = \pi^2 s_1$, which is regarded as a half-translation because of the identity
\begin{equation}
R^2 = T_1.
\end{equation}
The action of $R$ on the parameters is given by 
\begin{equation}
R.(a_0, a_1, a_2, c) = (a_2 a_0, q^{-1} a_2 a_1 , q a_2^{-1} , c),
\end{equation}
which is not translational, meaning that the action on variables does not directly give a $q$-difference equation as in the case of $T_1$. However, restricting to the parameter subspace on which $a_2 = q^{1/2}$, the action becomes
\begin{equation}
R.(a_0, a_1 , c) = (q^{1/2} a_0, q^{-1/2} a_1, c).
\end{equation}
With the parameter specialisation, the action of $R$ on the variables $f_0, f_1$ is given by 
\begin{equation}
R (f_1) = f_0, \quad \quad R (f_0) = \frac{q c^2}{f_0 f_1} \frac{1 + a_0 f_0}{a_0 + f_0},
\end{equation}
which, because of the translational motion in parameter space, induces the $q$-difference equation \eqref{qP2}, where $F_n = R^n(f_0)$, which is said to be a \emph{projective reduction} of equation \eqref{qP3}. \\

We now make some remarks about this reduction, as well as the restriction to the parameter subspace. In the language of Sakai's theory (to be established in section 2), the system \eqref{qP3} lifts to a family of birational isomorphisms $X_{R^n.\mathbf{a}} \longrightarrow X_{R^{n+1}.\mathbf{a}}$, where the sequence $\left\{ X_{R^n.\mathbf{a}} ~|~n\in\Z\right\}$ is a family of $A_5^{(1)}$-surfaces, with $R^n.\mathbf{a} = R^n.(a_0, a_1,a_2 , c)$ being the result of acting on the parameters $n$ times. The parameters $\mathbf{a} = (a_0, a_1, a_2 , c)$ are given by the \emph{root variables} associated with a basis of the symmetry root lattice of the family. In particular, this means that the action of the symmetry group $\widetilde{W}((A_2+A_1)^{(1)})$ on the parameters will correspond in a natural way to its action by Cremona isometries on the simple root basis of the symmetry lattice $Q((A_2+A_1)^{(1)})$. This guarantees that a translation element of the symmetry group, such as $T_1$, will give a translational motion in parameter space with root variables as coordinates. With regards to the element $R$, we note that without the restriction to the parameter subspace, its birational action on parameters and variables still defines a discrete dynamical system, as $R$ is of infinite order. The fact that the motion in parameter space is not translational for general $\mathbf{a}$ means that without the parameter constraint the system will not be given directly by a $q$-difference equation using these parameters, as we will illustrate when we construct it explicitly in section 5. \\

With regards to the parameter constraint, we note that the concept of translational motion is coordinate-dependent, in the sense that an automorphism of parameter space may be translational in one coordinate system but not another. In the geometric framework, we have a canonical choice of parametrisation for each family of surfaces by root variables, which ensures that the action on these parameters of any translation element of the symmetry group will be translational, and therefore the discrete Painlev\'e equations defined by Sakai are indeed difference equations of additive, multiplicative or elliptic type. The fact that the Cremona actions of non-translation elements of infinite order still give integrable difference equations with the full number of parameters for their surface type is the central idea of the second part of the paper.

\subsection{Outline of the paper}
The paper is structured as follows. In section 2 we introduce notation and recall material from Sakai's theory in order to formulate rigorously our geometric treatment of equation \eqref{hoffmans}. In section 3 we construct its space of initial values, compute bases for the surface and symmetry root lattices and obtain the induced Cremona isometry. We then express this Cremona isometry in terms of generators of the extended affine Weyl group, and show that it is the composition of a Kac translation and a Dynkin diagram automorphism. The main result of section 4 is the construction of a generic (5-parameter) version of equation \eqref{hoffmans}. To do this, we first introduce a 5-parameter family of $D_4^{(1)}$-surfaces generalising the space of initial values for equation \eqref{hoffmans}, and obtain the Cremona action of the extended affine Weyl group $\widetilde{W}(D_4^{(1)})$ on this family, from which we recover \eqref{hoffmans} as a special case of a projective reduction. Next we obtain the parametrisation of this family by root variables, through a transformation to the form of the family of generic $D_4^{(1)}$-surfaces given in \cite{KNY2017}. With the root variable parametrisation, we use the Cremona action to construct a generic version of equation \eqref{hoffmans}, and demonstrate that it has vanishing algebraic entropy. In section 5, we construct more examples of difference equations with the maximum number of free parameters for their surface type from non-translation elements of infinite order in the symmetry group, which are shown to be integrable but not directly of additive, multiplicative or elliptic type.

\section{Preliminaries on the geometric framework}

In this section we will introduce notation and highlight some important results in preparation for our geometric treatment of equation \eqref{hoffmans}. We concentrate on summarising the theory and results of Sakai rather than providing a review of rational surfaces and affine root systems. For a more comprehensive introduction to the background for the study of discrete Painlev\'e equations, we refer the reader to one of the excellent existing treatments, e.g. \cite{KNY2017,JNS2016,DT2018,SAKAI2001}.

\subsection{Root system of type $E_8^{(1)}$ in the Picard group}
Consider $\C^2$ with coordinates $(f,g)$. Compactify this to $\p^1 \times \p^1$, with $f$ and $g$ being local affine coordinates in each $\p^1$-factor. Then consider the smooth projective surface $X$ obtained by blowing up eight (possibly infinitely near) points $p_1, \dots , p_8$ on $\p^1 \times \p^1$.\\

We denote by $\text{Pic}(X)$ the Picard group, whose elements are line bundles on $X$, with the group operation being the tensor product. As $X$ is smooth, $\Pic(X)$ is isomorphic to the divisor class group, which is the quotient of the group $\text{Div}(X)$ of Weil divisors by linear equivalence. 
Writing operations additively, $\Pic(X)$ is a free $\Z$-module of rank 10 given by
\begin{equation}
\text{Pic}(X) = \Z \mathcal{H}_f \oplus \Z \mathcal{H}_g \oplus \Z \mathcal{E}_1 \oplus \dots \oplus \Z \mathcal{E}_8,
\end{equation}
where $\mathcal{H}_f$ and $\mathcal{H}_g$ are the total transforms of divisor classes of lines of constant $f$ and $g$ respectively, while $\mathcal{E}_i, ~i=1, \dots , 8$ are the exceptional classes arising from the eight blowups. 
 \\
The intersection product on $\text{Pic}(X)$ is the symmetric bilinear pairing given by extension of the formulae
\begin{equation}
\mathcal{H}_f \cdot \mathcal{H}_f = \mathcal{H}_g \cdot \mathcal{H}_g = \mathcal{H}_f \cdot \mathcal{E}_i = \mathcal{H}_g \cdot \mathcal{E}_j = 0, \quad \mathcal{H}_f \cdot \mathcal{H}_g = 1, \quad \mathcal{E}_i \cdot \mathcal{E}_j = - \delta_{ij},
\end{equation}
where $i, j \in \left\{1, \dots, 8 \right\}$. The top wedge product of the cotangent bundle on $X$ is the canonical bundle, which is given in terms of generators by $\mathcal{K}_X = -2 \mathcal{H}_f -2 \mathcal{H}_g + \mathcal{E}_1 + \dots + \mathcal{E}_8$. The dual of this is the anticanonical bundle 
\begin{equation}
-\mathcal{K}_X= 2\mathcal{H}_f + 2\mathcal{H}_g - \mathcal{E}_1 - \E_2-\E_3-\E_4-\E_5-\E_6-\E_7-\E_8,
\end{equation}
which corresponds to the equivalence class of the pole divisors of rational $2$-forms on $X$.\\

We will be concerned with rational surfaces which admit root system structures in $\Pic(X)$. For a root system of type $R$ (usually referred to by its Dynkin diagram), we denote the root lattice by $Q(R)$, the weight lattice by $P(R)$, and in the affine case, the null root by $\delta \in Q(R)$. We quote the following observation of Sakai \cite{SAKAI2001}:
\begin{proposition}
For $X$ as above, the Picard group equipped with the negative of the intersection pairing is isomorphic, as a free $\Z$-module with symmetric bilinear product, to the Lorentzian lattice of rank $10$. Further, the orthogonal complement in $\Pic(X)$ of the canonical class is isomorphic  to the root lattice $Q(E_8^{(1)})$, with the null root $\delta \in Q(E_8^{(1)})$ identified with the anticanonical class $- \mathcal{K}_X$. 
\end{proposition}

\subsection{Generalised Halphen surfaces of index zero}
Sakai defined a class of complex projective surfaces generalising Okamoto's space of initial values, and classified such surfaces into 22 types. Sakai defined a \emph{generalised Halphen surface of index zero} to be a rational surface whose anticanonical class is effective and of dimension zero, with the unique representative $D \in |-\mathcal{K}_X |$ being of canonical type. That is, its decomposition into irreducible components
\begin{equation}
D = \sum_{i \in I} m_i D_i, 
\end{equation}
is such that $- \mathcal{K}_X \cdot \delta_i = 0$ for all $i$, where $\delta_i = [D_i] \in \text{Pic}(X)$ are the divisor classes of the irreducible components. This definition leads to two important complementary root sublattices of $Q(E_8^{(1)}) \subset \Pic(X)$, as shown by the following result of Sakai \cite{SAKAI2001}.
\begin{proposition}
For a generalised Halphen surface $X$ of index zero, the irreducible components of the unique anticanonical divisor define a basis of simple roots $\left\{\delta_i \right\}_{i \in I}$ for an indecomposable root system in $\text{Pic}(X)$, with generalised Cartan matrix $A = (a_{ij})$ given by
\begin{equation}
a_{ij} = - \delta_i \cdot \delta_j .
\end{equation}
This root system is of affine type, with null root identified with the anticanonical class: 
\begin{equation}
\delta = - \mathcal{K}_X = \sum_{i \in I} m_i \delta_i.
\end{equation}
Further, if we denote the associated root lattice by $Q(R) = \bigoplus_{i \in I} \Z \mathcal{D}_i \subset Q(E_8^{(1)})$, the orthogonal complement $Q(R)^{\perp} \subset Q(E_8^{(1)})$ is another root lattice $Q(R^{\perp})$, which is also of affine type. The possible pairs of complementary root systems $R, R^{\perp}$ are classified according to indecomposable root subsystems of $E_8^{(1)}$. The possible surface root system types $R$ are shown in Figure 1, with arrows indicating degenerations of point configurations through which certain surface types may be obtained from others.
\end{proposition}
\begin{figure}[h] \label{Rtable}
\begin{center}
\begin{tikzcd}[row sep=small, column sep=small]
& &  &  &  &  & &  A_7^{(1)} \arrow[dr] \\
A_0^{(1)} \arrow[r] & A_1^{(1)} \arrow[r] & A_2^{(1)} \arrow[r] & A_3^{(1)} \arrow[r] \arrow[dr] & A_4^{(1)} \arrow[r] \arrow[dr] & A_5^{(1)} \arrow[r] \arrow[dr] \arrow[ddr] & A_6^{(1)} \arrow[r] \arrow[dr] \arrow[ur] & A_7^{(1)'} \arrow[bend left, dd] \arrow[dr] & A_8^{(1)} \arrow[bend left, dd] \\
& &  &  & D_4^{(1)} \arrow[r] & D_5^{(1)} \arrow[r] \arrow[dr] & D_6^{(1)} \arrow[r]  \arrow[dr] & D_7^{(1)} \arrow[r] & D_8^{(1)} \\
& &  &  &  &  & E_6^{(1)} \arrow[r] & E_7^{(1)}\arrow[r] & E_8^{(1)}
\end{tikzcd}
\end{center}
\caption{Surface root system type $R$ for generalised Halphen surfaces of index zero}
\end{figure}
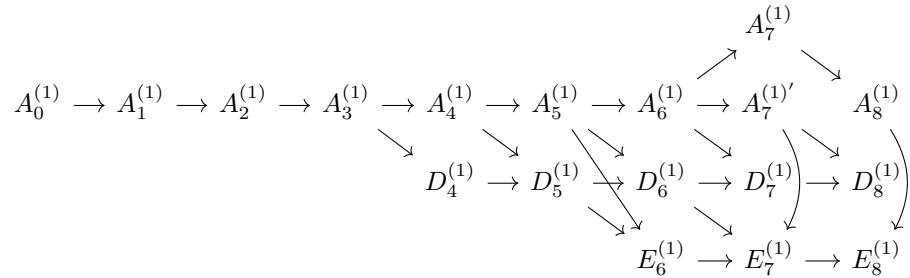
The classification given by Sakai of generalised Halphen surfaces of index zero is finer than the type $R$ of the root system defined by the components of $D$. By further classifying $X$ with unique anticanonical divisor $D = \sum_i m_i D_i$ according to $\text{rank} H_1\left( \sum_i D_i,\Z \right)$, we differentiate between families of surfaces whose associated discrete Painlev\'e equations are of elliptic, multiplicative and additive type. The 22 possible surface types are shown in Table 1, and we refer to them by type as $\mathcal{R}$-surfaces.

\renewcommand{\arraystretch}{1.5}
\begin{center}
\begin{table}[h] \label{surfacetable}
\begin{tabular}{ c | >{$}c<{$} }
Elliptic type & A_0^{(1)} \\ \hline   
Multiplicative type & A_0^{(1)*}, A_1^{(1)}, A_2^{(1)}, \dots, A_7^{(1)}, A_7^{(1)'}, A_8^{(1)} \\ \hline   
Additive type & A_0^{(1)**} , A_1^{(1)*}, A_2^{(1)*}, D_4^{(1)} , \dots, D_8^{(1)} \\
~ & E_6^{(1)} , E_7^{(1)}, E_8^{(1)} 
\end{tabular}
\caption{Classification of generalised Halphen surfaces index zero by surface type $\mathcal{R}$}
\label{GHsurfaces}
\end{table}
\end{center}
\renewcommand{\arraystretch}{1}
Surface types can be described in terms of configurations of nine points in $\p^2$, or equivalently eight points in $\p^1 \times \p^1$, which after blowups will give an anticanonical divisor with the required decomposition. For example, the elliptic $A_0^{(1)}$-surface may be obtained by blowing up eight points in general position on an irreducible elliptic curve in $\p^1 \times \p^1$ of bi-degree $(2,2)$, with the anticanonical divisor given by the proper transform of the curve. The most general form of such a surface (up to M\"obius transformation in each $\p^1$-factor) has eight free parameters $b_1, \dots b_8$ controlling the locations of the basepoints. Indeed, all but the most degenerate point configuration corresponding to surface type $E_8^{(1)}$ involve one or more free parameters, and we call a family of $\mathcal{R}$-surfaces with the maximum number of parameters for their type a \emph{family of generic }$\mathcal{R}$\emph{-surfaces}. \\

For any family $X_{\mathbf{b}}$ of $\mathcal{R}$-surfaces indexed by some list $\mathbf{b}$ of parameters, we may naturally identify their Picard groups and components of their anticanonical divisors to form a single $\Z$-module, which we also denote $\Pic(X)$. Special automorphisms \cite{DOLGACHEV1983, LOOIJENGA1981} of this $\Z$-module will correspond to symmetries of the family of surfaces, which we now describe.

\begin{definition}
Let $X_{\mathbf{b}}$ be a family of $\mathcal{R}$-surfaces, and let $\Pic(X)$ be the identification of the Picard groups $\Pic(X_{\mathbf{b}})$ as above. A linear map 
\begin{equation}
\varphi : \Pic(X) \rightarrow \Pic(X)
\end{equation}
is called a \emph{Cremona isometry} of the family of $\mathcal{R}$-surfaces if it:
\begin{enumerate}
\item preserves the intersection form on $\text{Pic}(X)$,
\item leaves the canonical class $\mathcal{K}_X$ fixed,
\item preserves effectiveness of divisor classes.
\end{enumerate}
\end{definition}
For each surface type $\mathcal{R}$, Sakai described the group $\text{Cr}(X(\mathcal{R}))$ of Cremona isometries of a family of generic $\mathcal{R}$-surfaces. We quote the part of this result relevant to this paper:
\begin{theorem}
For $\mathcal{R} \neq A_6^{(1)}, A_7^{(1)}, A_7^{(1)'}, A_8^{(1)}, D_7^{(1)}$ or $D_8^{(1)}$, the group of Dynkin diagram automorphisms $\text{Aut}(\mathcal{R})$ acts on $\Pic(X)$, and 
\begin{equation}
\text{Cr}(X(\mathcal{R})) \cong \left( \widetilde{W}(R^{\perp}) \right)_{\Delta^{\text{nod}}}.
\end{equation}
Here the right-hand side is the part of the extended affine Weyl group of the symmetry root system $R$ which stabilises the set of classes of nodal curves $\Delta^{\text{nod}} \subset \Pic(X)$. 
\end{theorem}
Thus we have an action of $\widetilde{W}(R^{\perp})$ on $\Pic(X)$ defined as the extension of the usual one on the root lattice $Q(R^{\perp})$, where the reflection associated to the root $\alpha \in R^{\perp}$ is given by 
\begin{equation}
r_{\alpha}(\lambda) = \lambda - 2\frac{\alpha \cdot \lambda}{\alpha \cdot \alpha} \alpha,
\end{equation}
where $\lambda \in \Pic(X)$, and Dynkin diagram automorphisms act on $Q(R^{\perp})$ in the natural way induced by the permutation of simple roots.

\subsection{Cremona action and discrete Painlev\'e equations}

In Sakai's theory, discrete Painlev\'e equations are constructed through a \emph{Cremona action} of the symmetry group, which realises the Cremona isometries as pullbacks of birational maps between specific members of a family of $\mathcal{R}$-surfaces. The fact that the birational maps are between different elements of the family of surfaces is essential for obtaining non-autonomous difference equations. \\

The Cremona action is most conveniently described using the root variable parametrisation of the family of generic $\mathcal{R}$-surfaces, in which the parameters are associated to simple roots in $Q(R^{\perp}) \subset \Pic(X)$, as we now illustrate.
\begin{definition}
Consider a generalised Halphen surface $X$ with unique anticanonical divisor $D = \sum_i m_i D_i$, and with surface and symmetry root systems of type $R$ and $R^{\perp}$ respectively. Let $\omega$ be a rational $2$-form on $X$ such that $\div (\omega)=-D$, and denote $D_{\text{\emph{red}}} = \sum_i D_i$. Then from relative homology of the pair $(X , X - D_{\text{\emph{red}}})$ and Poincar\'e duality on $D_{\text{\emph{red}}}$ we have a short exact sequence:
\begin{equation}
0 \longrightarrow H_1(D_{\text{\emph{red}}} ; \Z) \longrightarrow H_2(X-D_{\text{\emph{red}}} ; \Z) \longrightarrow Q(R^{\perp}) \longrightarrow 0.
\end{equation}
The isomorphism given by this  sequence, together with the map
\begin{equation}
\begin{aligned}
\hat{\chi} : &~ H_2(X -D_{\text{\emph{red}}} ; \Z) \rightarrow \C, \\
&\quad \Gamma \mapsto \int_{\Gamma} \omega,
\end{aligned}
\end{equation}
defines the period mapping
\begin{equation}
\chi : Q(R^{\perp}) \rightarrow \C \mod \hat{\chi}( H_1(D_{\text{\emph{red}}} ; \Z)).
\end{equation}
\end{definition}
In practice, we choose the rational $2$-form $\omega$ according to a normalisation such that the period mapping is a well-defined function. The kind of normalisation required depends on the rank of $H_1(D_{\text{red}} ; \Z)$, and therefore whether the $\mathcal{R}$-surface is of elliptic, multiplicative or additive type. The period mapping allows us to construct the \emph{root variable parametrisation} of a family of generic $\mathcal{R}$-surfaces. Pick a basis of simple roots $\left\{\alpha_0, \dots, \alpha_l \right\} \subset Q(R^{\perp})$ and define the \emph{simple root variables} as
\begin{equation}
a_j = \chi (\alpha_j).
\end{equation}
Together with the `extra parameter' corresponding to the independent variable of the continuous Painlev\'e equation in the cases $\mathcal{R} = D_4^{(1)}, D_5^{(1)}, D_6^{(1)}, D_7^{(1)}, D_8^{(1)}, E_6^{(1)}, E_7^{(1)}$ and $E_8^{(1)}$, the root variables allow us to parametrise the family of generic $\mathcal{R}$-surfaces.\\

We are now ready to describe the Cremona action of the group $\widetilde{W}(R^{\perp})$. This consists of birational maps constructed from changes of blowing-down structures of generic $\mathcal{R}$-surfaces which induce Cremona isometries through their pullbacks. Such a change of blowing-down structure induces a change in the root variables of the surface and therefore a change in parameters in the form of a generic $\mathcal{R}$-surface, so can be regarded as a birational map between different members of a family of surfaces. Thus the Cremona action can be described as an action of the extended affine Weyl group $\widetilde{W}(R^{\perp})$ on the parameters $\mathbf{a} = \left\{a_0, \dots, a_l \right\}$ and local affine coordinates $(f,g)$ for $X_{\mathbf{a}}$, which play the role of variables in the discrete Painlev\'e equations. This action defines birational maps
\begin{equation}
\begin{aligned}
&w :~ X_{\mathbf{a}} \rightarrow X_{w. \mathbf{a}}, \\
&~~~(f,g) \mapsto (w.f, w.g),
\end{aligned}
\end{equation}
which give the actions of $\widetilde{W}(R^{\perp})$ by Cremona isometries by their pullbacks, when the Picard groups of the surfaces are identified with the module $\Pic(X)$.\\


Sakai's discrete Painlev\'e equations arise from the translation part of $\widetilde{W}(R^{\perp})$, which we now describe. As $R^{\perp}$ is of affine type, the root system obtained by deleting the $0$th row and column from the generalised Cartan matrix is of finite type, which we denote $R_{\circ}^{\perp}$. The affine Weyl group $W(R^{\perp})$ (not extended by the Dynkin diagram automorphisms $\text{Aut}(R^{\perp})$) can be expressed as the semi-direct product
\begin{equation}
W(R^{\perp}) \cong W(R_{\circ}^{\perp}) \ltimes Q(R_{\circ}^{\perp}).
\end{equation}
The translation part of the extended affine Weyl group comes similarly from the semi-direct product
\begin{equation}
\widetilde{W}(R^{\perp}) \cong W(R_{\circ}^{\perp}) \ltimes P(R_{\circ}^{\perp}),
\end{equation}
where the weight lattice $P(R_{\circ}^{\perp})$ gives Cremona isometries via \emph{Kac translations} \cite{KAC}, according to the formula
\begin{equation} \label{kactrans}
T_{v} (\lambda) = \lambda - (\lambda \cdot \delta) v - \left( \frac{1}{2} (v \cdot v) (\lambda \cdot \delta) - \lambda \cdot v \right) \delta,
\end{equation}
for any $\lambda \in \Pic(X)$, where $v \in P(R_{\circ}^{\perp})$ is the translation vector. We are now ready to formally define Sakai's discrete Painlev\'e equations.
\begin{definition}
A discrete Painlev\'e equation of surface type $\mathcal{R}$ is a second-order difference equation given by the Cremona action of a Kac translation $T_v$, for some $v \in P(R_{\circ}^{\perp})$.

\end{definition}
In other words, we call a second-order difference equation a discrete Painlev\'e equation if it lifts to a family of birational isomorphisms between $\mathcal{R}$-surfaces given by the Cremona action of a Kac translation by an element of the weight lattice associated with the symmetry root system.\\

\section{Space of initial values and Cremona isometry}

We now give a geometric treatment of equation \eqref{hoffmans}, beginning with the construction of its space of initial values.

\begin{proposition}
The system \eqref{hoffmans} lifts to a family of birational isomorphisms 
\begin{equation}
\Phi_{n} : X_{(k,n)} \longrightarrow X_{(k,n+1)},
\end{equation}
where for each $n\in \Z$ the surface $X_{(k,n)}$ is obtained from $\p^1 \times \p^1$ by blowing up the points $p_1, \dots, p_8$ given in coordinates by
\begin{subequations} \label{pointlocations}
\begin{align}
p_1 : (f,g) &= (k,0), &&\quad p_5 : (u_1,v_1) = \left(\frac{f-k}{g} , g \right) = \left( (k^2-1)(2n+1), 0 \right), \\  
p_2 : (F,G) &= (k, 0), &&\quad p_6 : (u_2,v_2) = \left(\frac{F-k}{G} , G \right) = \left( (k^2-1)(2n+1), 0 \right), \\  
p_3 : (f,g) &= (0,k), &&\quad p_7 : (u_3,v_3) = \left(\frac{f}{g-k} , g-k \right) = \left( 0, 0 \right), \\  
p_2 : (F,G) &= (0, k), &&\quad p_8 : (u_4,v_4) = \left(\frac{F}{G-k} , G-k \right) = \left( 0, 0 \right).
\end{align}
\end{subequations} 
\end{proposition}

\begin{proof}

We compactify the domain and target spaces of the forward iteration map from $\C^2$ to $\p^1 \times \p^1$. We cover this with charts $(f,g), (F,g), (f,G)$ and $(F,G)$ for the domain space, where $F = 1/f, G=1/g$, and the same charts with the overline notation for the target space. We will use blowups to resolve the points of indeterminacy in $\p^1 \times \p^1$ of the forward and backward mappings, which we call \emph{basepoints} by analogy with the case of linear systems of divisors. Lifting the maps defined by iteration of the system to the blown-up surfaces, we will obtain a family of birational isomorphisms. \\
Recall that in the $(f,g)$ chart, we have the forward iteration of the system given by the mapping $(f,g) \mapsto (\bar{f},\bar{g})$, where
\begin{equation} \label{fwd}
\bar{f} = \frac{(k-\bar{g})^2}{f (k \bar{g}-1)^2}, \quad \quad \bar{g} = \frac{(k-f) \left( k fg - (2n+1) f - g +k(2n+1) \right)}{(kf-1) \left( k(2n+1) f g- f - (2n+1) g + k \right)}, 
\end{equation}
whereas the backward iteration is given by the inverse of this, which we write as 
\begin{equation}\label{back}
f = \frac{(k- \bar{g})^2}{\bar{f} (k \bar{g}-1)^2}, \quad \quad g = \frac{(f - k)\left( k f \bar{g} + (2n+1) f - \bar{g} - k(2n+1) \right)}{(k f - 1) \left( k(2n+1) f \bar{g} + f - (2n+1) \bar{g} - k \right)}.
\end{equation}
Direct calculation shows that the indeterminacies of the forward mapping are
\begin{equation}
p_1 : (f,g) = (k,0), \quad \quad p_2 : (F,G) = (k, 0),
\end{equation}
while the indeterminacies of the backward mapping are given in coordinates by
\begin{equation}
\bar{p}_3 : (\bar{f},\bar{g}) = (0,k), \quad \quad \bar{p}_4 : (\bar{F}, \bar{G}) = (0, k).
\end{equation}
We blow up these points in both the domain and target copies of $\p^1\times \p^1$. For each $j=1, \dots, 4$, the exceptional divisor $E_j$ replacing $p_j$ in the domain space is covered by the pair of local affine coordinate charts $(u_j, v_j)$ and $(U_j, V_j)$, with the part of the line visible in the first chart parametrised by $u_j$ when $v_j = 0$, and similarly in the second chart by $U_j$ when $V_j = 0$. The basepoints $p_j : (f,g) = (f_j, g_j)$ for $j=1,3$ are visible in the $(f,g)$ chart, and the blowup coordinates are defined as
\begin{equation}
u_j = \frac{ f - f_j}{g- g_j}, \quad v_j = g- g_j, \quad \quad U_j = \frac{g-g_j}{f - f_j}, \quad V_j = f- f_j,
\end{equation}
and similarly
\begin{equation}
u_j = \frac{ F - F_j}{G- G_j}, \quad v_j = G- G_j, \quad \quad U_j = \frac{G-G_j}{F - F_j}, \quad V_j = F- F_j,
\end{equation}
for the points $p_j : (F,G) = (F_j, G_j)$, for $j=2,4$. \\
In order to examine the image under the forward mapping of the exceptional line $E_1$ that replaces $p_1$ after the blowup, we rewrite \eqref{fwd} using the chart $(u_1, v_1)$ for the domain, and $(\bar{f},\bar{g})$ for the target. Direct calculation reveals another basepoint on the exceptional line $E_1$, with the forward mapping becoming indeterminate at
\begin{equation}
p_5 : (u_1,v_1) = \left(\frac{f-k}{g} , g \right) = \left(  (k^2-1)(2n+1), 0 \right).
\end{equation} 
The corresponding point $\bar{p}_5$ in the target space is also an indeterminacy of the backward iteration, which is evidenced by the following. Writing the mapping in charts $(f,g)$ and $(\bar{u}_1, \bar{v}_1)$, we see that $f=k, g \neq 0$ implies 
\begin{equation}
\bar{u}_1 = (k^2-1)(2n+3), \quad \quad \bar{v}_1 = 0,
\end{equation}
so the forward mapping sends the line $f=k$ (excluding $p_1$) to the single point $\bar{p}_5$ on $\bar{E}_1$. We blow up this point in both the domain and target spaces too, covering the exceptional line $E_5$ with the affine charts $(u_5,v_5)$ and $(U_5,V_5)$ defined by 
\begin{subequations}
\begin{align}
&u_5 = \frac{ u_1 - (k^2-1)(2n+1)}{v_1},  &&v_5 = v_1, \\
&U_5 = \frac{v_1}{ u_1 - (k^2-1)(2n+1)}, &&V_5 = { u_1 - (k^2-1)(2n+1)}.
\end{align}
\end{subequations}
Writing \eqref{fwd} in charts $(f,g)$ and $(\bar{u}_5, \bar{v}_5)$, we see that $f=k$ with $g \neq 0$ implies that 
\begin{equation} \label{E5}
\bar{u}_5 = \frac{ (k^2-1) \left( (k^2(4n+5)-4(n+1)^2) g - 4 n (2n^2 + 3n+1) k \right)}{k g}, \quad \bar{v}_5 = 0,
\end{equation}
so we have a one-to-one correspondence under the forward mapping between the proper transform of the line $f=k$ in the domain surface and the exceptional line $\bar{E}_5$ in the target surface. \\
Next, we compute the image of the exceptional line $E_1$ under the forward mapping, using the charts $(u_1, v_1)$ and $(\bar{u}_1, \bar{v}_1)$ we see that when $v_1=0$ and $u_1 \neq (k^2-1)(2n+1)$, we have
\begin{equation} \label{E1}
\bar{u}_1 = \frac{ (k^2-1)\left( (k^2-1)(2n+3) - (4n+3) u_1\right)}{k^2-1 - (2n+1) u_1}, \quad \bar{v}_1 = 0.
\end{equation}
With similar results in the charts $(U_1, V_1)$, we see that the forward mapping injects $E_1$ (excluding $p_5$) into $\bar{E}_1$. We now examine the image of the exceptional line $E_5$ under the forward mapping, writing \eqref{fwd} in charts $(u_5, v_5)$ and $(\bar{f}, \bar{g})$ then setting $v_5 = 0$, we see that the exceptional line $E_5$ is mapped bijectively to the proper transform of the curve defined in the $(\bar{f},\bar{g})$ chart by
\begin{equation} \label{E5imagecurve}
k \bar{f} ( k \bar{g}-1)^2 - (k - \bar{g})^2 = 0.
\end{equation}
We also deduce from \eqref{E1} that the proper transforms $E_1-E_5$ and $\bar{E}_1 - \bar{E}_5$ of the exceptional lines arising from $p_1, \bar{p}_1$ are in bijective correspondence under the mapping. Further, \eqref{E5} shows that the exceptional line $\bar{E}_5$ is in bijective correspondence with the proper transform $H_f - E_1$ of the line $f=k$, so we have lifted the mapping to a birational isomorphism in the neighbourhood of $p_1$. Calculations in the neighbourhoods of the exceptional lines $E_2, E_3$ and $E_4$ reveal similar phenomena, such as $H_f- E_2$ being blown down by the forward mapping to the single point $p_6 \in E_2$, while under the backward mapping $H_{\bar{g}} - \bar{E_3}$ and $H_{\bar{g}} - \bar{E_4}$ are blown down to $p_7 \in E_3$ and $p_8 \in E_4$ respectively. After blowing up these points $p_6, p_7$ and $p_8$, calculations in local coordinates similar to those above show that the system lifts to a family of birational isomorphism as claimed.


\end{proof}
We next show that the surfaces $X_{(k,n)}$ are a family of generalised Halphen surfaces of index zero, and place them in Sakai's classification.

\begin{proposition}
Each surface $X_{(k,n)}$ has a unique representative of its anticanonical class, given by 
\begin{equation}
D = D_0 + D_1 + 2 D_2 + D_3 + D_4,
\end{equation}
where the divisor classes $\delta_j = [D_j]$ of the irreducible components are 
\begin{subequations} \label{hoffmanssurfaceroots}
\begin{align}
\delta_0 &= \mathcal{E}_1-\mathcal{E}_5, \quad \delta_1 = \mathcal{E}_2-\mathcal{E}_6, \quad \delta_2 = \mathcal{H}_f + \mathcal{H}_g - \mathcal{E}_1- \mathcal{E}_2- \mathcal{E}_3- \mathcal{E}_4, \\ 
\delta_3 &= \mathcal{E}_3-\mathcal{E}_7, \quad \delta_4 = \mathcal{E}_4-\mathcal{E}_8.
\end{align}
\end{subequations}
The divisor $D$ is of canonical type, and $X_{(k,n)}$ form a family of $\mathcal{R}$-surfaces with $\mathcal{R} = D_4^{(1)}$. The surface root lattice is given by $Q(R) = \text{span}_{\Z}\left\{ \delta_0 , \dots , \delta_4 \right\}$, while the symmetry root lattice is $Q(R^{\perp}) = \text{span}_{\Z}\left\{\alpha_0, \dots \alpha_4 \right\}$, with the simple roots given by
\begin{subequations} \label{hoffmanssymmetryroots}
\begin{align}
\alpha_0 &= \H_g -\mathcal{E}_1 -\mathcal{E}_5,  \quad \alpha_1 = \H_g -\mathcal{E}_2 -\mathcal{E}_6, \quad \alpha_2 = \H_f - \mathcal{H}_g, \\
\alpha_3 &= \H_g -\mathcal{E}_3 -\mathcal{E}_7, \quad \alpha_4 = \H_g -\mathcal{E}_4 -\mathcal{E}_8.
\end{align}
\end{subequations}
\end{proposition}
\begin{proof}
We aim to identify irreducible divisors of $X_{(k,n)}$ which give the decomposition of a unique representative $D \in | - \mathcal{K}_{X_{(k,n)}} | = 2 \H_f + 2\H_g - \E_1 - \dots -\E_8 \in \Pic(X_{(k,n)})$. We first note that the proper transforms of the exceptional lines $E_1, \dots, E_4$ under the blowups of $p_5, \dots , p_8$ are the divisors $E_1 - E_5,~ E_2-E_6,~ E_3- E_7$ and $E_4- E_8$ respectively. These have self-intersection $-2$ and are irreducible. Further, the proper transform of the curve defined in the chart $(f,g)$ by
\begin{equation} \label{11curve}
k f g - f - g - k = 0,
\end{equation}
corresponds to the divisor $H_f + H_g - E_1 -E_2 - E_3 -E_4$. This can be verified by noting that the polynomial defining the curve is of bi-degree $(1,1)$ in $(f,g)$, and the curve passes through the points $p_1, \dots , p_4$ with multiplicity 1, but its proper transform under their blowups does not intersect $p_5, \dots, p_8$, which can be checked by direct calculation. We deduce that this curve defines a representative of the divisor class $\H_f + \H_g - \E_1 -\E_2 - \E_3 - \E_4$, and compute its self-intersection to be $-2$. If we denote these five divisors by 
\begin{subequations} \label{components}
\begin{align}
D_0 &= E_1-E_5, \quad D_1 = E_2-E_6, \quad D_2 = H_f + H_g - E_1 -E_2 - E_3 -E_4, \\ 
D_3 &= E_3-E_7, \quad D_4 = E_4-E_8,
\end{align}
\end{subequations}
we see that $D = D_0 + D_1 + 2 D_2 + D_3 + D_4$ is a representative of the anticanonical class of $X_{(k,n)}$, as 
\begin{equation}
[D] = [D_0] + [D_1] + 2 [D_2] + [D_3] + [D_4] = 2\mathcal{H}_f + 2\mathcal{H}_g - \mathcal{E}_1 - \dots -\E_8 = - \mathcal{K}_{X_{(k,n)}}.
\end{equation}
With regards to the dimension of the anticanonical class, each of the divisors $D_0, D_1, D_3$ and $D_4$ defined as proper transforms of exceptional lines give unique representatives of their classes in $\Pic(X)$. Further, it can be shown by direct calculation that \eqref{11curve} defines the unique curve of bi-degree $(1,1)$ passing through $p_1, \dots , p_4$, which means the divisor $D_2$ does not move in a family and its class in $\Pic(X_{(k,n)})$ is of dimension zero. Thus we deduce that $D$ is the unique representative of the anticanonical class. Direct computation shows that $[D_j] \cdot \mathcal{K}_X = 0$ for $j=0, \dots ,4$, so we see also that $D$ is of canonical type. \\
Computing the pairwise intersections of the classes $\delta_j = [D_j]$, for $j=0,\dots, 4$, we obtain 
\begin{equation}
- ( \delta_i \cdot \delta_j ) = \left(\begin{array}{ccccc}2 & 0 & -1 & 0 & 0 \\0 & 2 & -1 & 0 & 0 \\-1 & -1 & 2 & -1 & -1 \\0 & 0 & -1 & 2 & 0 \\0 & 0 & -1 & 0 & 2\end{array}\right).
\end{equation}
The matrix on the right-hand side is the generalised Cartan matrix of type $D_4^{(1)}$, meaning that each $X_{(k,n)}$ is a $D_4^{(1)}$-surface in the sense of Sakai's theory. \\
The root lattice given by the orthogonal complement $Q(R)^{\perp} \subset \delta^{\perp} \subset \Pic(X)$ is also of type $R^{\perp} = D_4^{(1)}$. We find a basis of simple roots for $Q(R^{\perp})$ by looking for a set of five linearly independent elements of $\delta^{\perp}$ with the required intersection numbers. This yields a system of linear equations for the coefficients of the generators in each of the simple roots, from which we obtain the basis \eqref{hoffmanssymmetryroots} for the symmetry root lattice, and the proof is complete.
\end{proof}
We now identify the Cremona isometry induced by the family of birational isomorphisms, which determines whether the system is a discrete Painlev\'e equation according to definition 2.6.

\begin{theorem}
The family of birational maps $\Phi_n$ lifted from the system \eqref{hoffmans} induce via their pullbacks the map $\varphi : \text{Pic}(X) \rightarrow \text{Pic}(X)$ given by 
\begin{equation} \label{isometry}
\varphi : \left\{ \begin{array}{lcl}
\H_f & \mapsto & 5\H_f + 2 \H_g -2\E_1-2\E_2- \E_3 - \E_4 -2\E_5-2\E_6- \E_7 - \E_8, \\
\H_g & \mapsto & 2 \H_f + \H_g - \E_1 - \E_2  - \E_5 - \E_6, \\
\E_1 & \mapsto & \H_f - \E_5, \\
\E_2 & \mapsto & \H_f  - \E_6, \\
\E_3 & \mapsto &  2\H_f + \H_g - \E_1- \E_2 - \E_5 - \E_6 - \E_7, \\
\E_4 & \mapsto &  2\H_f + \H_g - \E_1- \E_2 - \E_5 - \E_6 - \E_8, \\
\E_5 & \mapsto &  \H_f - \E_1, \\
\E_6 & \mapsto &  \H_f - \E_2, \\
\E_7 & \mapsto &  2\H_f + \H_g - \E_1- \E_2 - \E_3 - \E_5 - \E_6, \\
\E_8 & \mapsto &  2\H_f + \H_g - \E_1- \E_2 - \E_4 - \E_5 - \E_6.
\end{array}
\right.
\end{equation}
In particular, $\varphi$ acts on the symmetry roots according to 
\begin{equation} \label{phionalpha}
\varphi : \left( \alpha_0, \alpha_1, \alpha_2, \alpha_3, \alpha_4 \right) \mapsto \left( \alpha_1, \alpha_0, \alpha_2, \alpha_4, \alpha_3 \right) + \left( 0 ,0 , \delta, -\delta, -\delta \right).
\end{equation}
Further, $\varphi$ is a Cremona isometry for the family $X_{(n,k)}$ of $D_4^{(1)}$-surfaces, and can be written in terms of generators of $\widetilde{W}(D_4^{(1)})$ as
\begin{equation} \label{generatorsexpression}
\varphi = r_{201234} = \pi_{(01)(34)} T_{v},
\end{equation}
where $T_v$ is Kac translation by the weight vector $v = \frac{1}{2} (\alpha_3+\alpha_4) \in P(R_{\circ}^{\perp})$, and $\pi_{(01)(34)}$ is the Cremona isometry given by the Dynkin diagram automorphism swapping $\alpha_0, \alpha_1$, and $\alpha_3, \alpha_4$.

\end{theorem}
\begin{proof}
We note that there are many equivalent calculations by which the map \eqref{isometry} may be deduced from the birational isomorphisms $\Phi_n$. The basic idea is to choose effective classes in $\Pic(X_{(k,n)})$ such that computing the images of their representatives under the pushforward 
\begin{equation}
(\Phi_{n})_* : \text{Div}(X_{(k,n)}) \rightarrow  \text{Div}(X_{(k,n+1)}),
\end{equation}
is as simple as possible, and gives enough conditions to deduce $\varphi$. The calculations we outline are based on those from the proof of Proposition 3.1. \\
First recall that by considering the forward mapping in charts $(u_1, v_1)$ and $(\bar{u}_1, \bar{v}_1)$, we obtained \eqref{E1}, from which we deduce that $(\Phi_n)_* (E_1 - E_5) = \bar{E}_1 - \bar{E}_5$, which means that with the identification of Picard groups of the family $X_{(k,n)}$ we have
\begin{equation} \label{E1E5}
\varphi^{-1}(\E_1 - \E_5) = \E_1 - \E_5.
\end{equation}
Similarly, our calculation of the image under $\Phi_n$ of the line $f=k$ led to \eqref{E5}, which shows that $(\Phi_n)_* (H_f-E_1) = \bar{E}_5$ and therefore
\begin{equation} \label{HfE1}
\varphi^{-1}(\H_f - \E_1) = \E_5.
\end{equation}
The map $\varphi^{-1}$ is linear by construction, so if we obtain $\varphi^{-1}(\E_5)$, we may deduce $\varphi^{-1}(\H_f)$ and $\varphi^{-1}(\E_1)$ from \eqref{E1E5} and \eqref{HfE1}. In order to do this, we consider the mapping in charts $(u_5, v_5)$ and $(\bar{f}, \bar{g})$ and obtain the image of $E_5$ to be the proper transform of the curve defined by \eqref{E5imagecurve}. To determine the element of $\Div(X_{(k,n+1)})$ defined by the proper transform, we check its intersection with the exceptional lines $\bar{E_1}, \dots \bar{E}_8$. For example, to check intersections with $\bar{E}_1$ and $\bar{E}_5$, we substitute $\bar{f} = \bar{u}_1 \bar{v}_1 + k$ and $\bar{g} = \bar{v}_1$ into the equation of the curve \eqref{E5imagecurve}, which gives the equation of the total transform in the chart $(\bar{u}_1, \bar{v}_1)$ as
\begin{equation}
\bar{v}_1 \left(  k ( k \bar{v}_1 - 1)^2\bar{u}_1 + (k^4-1)\bar{v}_1 - 2k(k^2-1) \right) = 0.
\end{equation}
The factor $\bar{v}_1$ appearing here with exponent one indicates that the proper transform of the curve intersects $\bar{E}_1$ with multiplicity one, so $(\Phi_n)_*(E_5) \cdot \bar{E}_1 = 1$. Further, the proper transform of the curve \eqref{E5imagecurve} is given by 
\begin{equation}
 k ( k \bar{v}_1 - 1)^2\bar{u}_1 + (k^4-1)\bar{v}_1 - 2k(k^2-1) = 0.
\end{equation}
The coordinates of the point $\bar{p}_5$ do not satisfy this equation, so we deduce that $(\Phi_n)_*(E_5) \cdot \bar{E}_5 = 0$, so the coefficient of $\E_5$ in $\varphi^{-1}(\E_5)$ is zero. Similar calculations in charts covering the rest of the exceptional lines, and the observation of the bi-degree of the polynomial defining the curve allow us to deduce the coefficients, and we arrive at
\begin{equation}
\varphi^{-1} (\E_5) = \H_f + 2\H_g - \E_1 - \E_3 - \E_4 - \E_7 - \E_8.
\end{equation}
By similar calculations we obtain more conditions on $\varphi$:
\begin{subequations}
\begin{align}
\varphi^{-1}(\E_2 - \E_6) &= E_2 - E_6, \quad \quad  \varphi^{-1}(\E_3 - \E_7) = E_3 - E_7, \quad \quad \varphi^{-1}(\E_4 - \E_8) = E_4 - E_8, \\
\varphi^{-1}(\H_g - \E_1) &= \H_f + 3\H_g -\E_1- \E_2 - \E_3 - \E_4 -\E_6 - \E_7 - \E_8, \\
\varphi^{-1}(\H_g - \E_2) &= \H_f + 3\H_g -\E_1- \E_2 - \E_3 - \E_4 -\E_5 - \E_7 - \E_8, \\
\varphi^{-1}(\E_6) &=  \H_f + 2\H_g - \E_2 - \E_3 - \E_4 - \E_7 - \E_8, \\
\varphi^{-1}(\E_7) &= \H_g -  \E_3, \quad \quad \varphi^{-1}(\E_8) = \H_g -  \E_4.
\end{align}
\end{subequations}
From these, we find the map $\varphi$ induced by the pullbacks to be given by \eqref{isometry}, from which we may check directly that $\varphi$ fixes each of $\delta_0 , \delta_1, \delta_2, \delta_3$ and $\delta_4$, as well as the canonical class in $\Pic(X)$, and that $\varphi$ preserves the intersection product. We also note that as $\varphi$ is induced by the pullbacks of birational isomorphisms it must preserve effectiveness of classes in $\Pic(X)$. Therefore $\varphi$ is a Cremona isometry of the family of $D_4^{(1)}$-surfaces, and by Sakai's results (theorem 2.4) we may express it in terms of generators of $\widetilde{W}(D_4^{(1)})$. \\
Recall that for each simple root $\alpha_j$ from our basis \eqref{hoffmanssymmetryroots} for $ Q(R^{\perp}) \subset \Pic(X)$, the corresponding simple reflection is defined by the formula
\begin{equation}
r_{j} (\lambda) = \lambda - 2\frac{\alpha_j \cdot \lambda}{\alpha_j \cdot \alpha_j} \alpha_j,
\end{equation} 
for $\lambda \in \Pic(X)$. A Dynkin diagram automorphism of $R^{\perp}$ will be denoted by $\pi_{\sigma}$, where $\sigma$ is the permutation of indices of nodes of the Dynkin diagram corresponding to the simple roots $\alpha_0, \dots ,\alpha_4$. The extended affine Weyl group $\widetilde{W}(D_4^{(1)})$ is then generated by these simple reflections together with a generating set for $\text{Aut}(R^{\perp})$, and by standard techniques \cite{KAC,KNY2017,DT2018}, we obtain the expression
\begin{equation}
\varphi = r_2 r_0 r_1 r_2 r_3 r_4 =  r_{201234},
\end{equation}
where for conciseness we have used the subscript notation to indicate the composition of a sequence of simple reflections. The action of the element $\varphi$ on the symmetry root lattice $Q(R^{\perp})$ is found to be given by \eqref{phionalpha} by direct calculation. In light of the Kac translation formula \eqref{kactrans}, we note that $\varphi$ is not itself a translation, otherwise it would act on the simple roots by adding multiples of the null root $\delta$. Consider the Dynkin diagram automorphism $\pi_{(01)(34)}$, which cycles each of the pairs $\alpha_0, \alpha_1$, and $\alpha_3, \alpha_4$ is realised by the following action on $\Pic(X)$:
\begin{equation}
\pi_{(01)(34)} : \quad  \E_1 \leftrightarrow \E_2, \quad \E_3 \leftrightarrow \E_4, \quad \E_5 \leftrightarrow \E_6, \quad \E_7 \leftrightarrow \E_8. 
\end{equation}
The composition of this with $\varphi$ acts on the symmetry root lattice according to
\begin{equation} \label{Tonalpha}
\pi_{(01)(34)} \varphi : \left( \alpha_0, \alpha_1, \alpha_2, \alpha_3, \alpha_4 \right) \mapsto \left( \alpha_0, \alpha_1, \alpha_2, \alpha_3, \alpha_4 \right) + \left( 0 ,0 , \delta, -\delta, -\delta \right),
\end{equation}
which is translational, as we now demonstrate. To determine the weight vector $v \in P(R_{\circ}^{\perp}) = P(D_4)$ giving this action via the Kac translation $T_v$, we form an ansatz $v = \sum_{j=1}^4 c_j \alpha_j$, and determine rational coefficients $c_j$ such that 
\begin{equation}
\left( v \cdot \alpha_0, v \cdot \alpha_1, v \cdot \alpha_2, v \cdot \alpha_3, v \cdot \alpha_4 \right) = (0,0,1, -1,-1),
\end{equation}
which ensures $T_v$ gives the action \eqref{Tonalpha} according to the Kac translation formula \eqref{kactrans}. By direct calculation in this case we obtain the unique solution $(c_1, c_2, c_3, c_4) = (0,0, 1/2, 1/2)$, so we have the translation vector
\begin{equation}
v = \frac{1}{2} \left(\alpha_3 + \alpha_4\right),
\end{equation}
and the proof is complete.
\end{proof}
\begin{remark}
The weight vector $v$ has squared length $||v||^2 =  - (v \cdot v) = 1$, which is minimal among all nonzero elements of $P(D_4)$, so we may think of $v$ as a `nearest-neighbour connecting vector' in the $D_4$ weight lattice, in the sense of \cite{RGO2001,JN2017}. Further, we have that 
\begin{equation}
\varphi^2 = T_v^2 = T_{2v},
\end{equation}
so in particular the element $\varphi = \pi_{(01)(34)} T_v$ is of infinite order in $\widetilde{W}(D_4^{(1)})$. Further, if we are to regard \eqref{hoffmans} as a projective reduction of a discrete Painlev\'e equation, it is the one given by the translation $T_{2v}$. This contrasts with the projective reduction introduced in the introduction, in the sense that the element $R \in \widetilde{W}( (A_2 + A_1)^{(1)})$ squared to give a translation by a nearest-neighbour connecting vector, which $T_{2v}$ is not.
\end{remark}

\section{Generic version of the equation}

We have found that the system \eqref{hoffmans} is regularised on a family of $D_4^{(1)}$-surfaces with less than the full number of parameters, and corresponds to a non-translation element of $\widetilde{W}(D_4^{(1)})$. In this section, we demonstrate first how the equation may be recovered as a kind of projective reduction using a birational representation of the symmetry group. Unlike previous studies of such systems, we proceed to show that although the action of the element of the symmetry group is not translational on the full parameter space, it still defines a difference equation with the full number of parameters, which we construct explicitly from the birational representation.

\subsection{Birational representation of $\widetilde{W}(D_4^{(1)})$ on generic $D_4^{(1)}$-surfaces}
The first step in our construction of a full-parameter version of the system \eqref{hoffmans} is to define a family of generic $D_4^{(1)}$-surfaces through a basepoint configuration generalising that which gives the family $X_{(k,n)}$. We introduce the extra parameters controlling basepoint locations first in a naive way, and give a birational representation of $\widetilde{W}(D_4^{(1)})$ on this family, with the system \eqref{hoffmans} recovered as a kind of projective reduction. We then obtain the root variable parametrisation in subsection 4.2, which allows our full-parameter version to be constructed explicitly. 
\begin{proposition}
Consider the surface $X_{\mathbf{b}} = X_{(k, b_1,b_2,b_3,b_4)}$ obtained from $\p^1 \times \p^1$ by blowing up points $p_1, \dots, p_8$ given in coordinates by
\begin{subequations}
\begin{align}
p_1 : (f,g) &= (k,0), &&\quad p_5 : (u_1,v_1) = \left(\frac{f-k}{g} , g \right) = \left( b_1, 0 \right), \\  
p_2 : (F,G) &= (k, 0), &&\quad p_6 : (u_2,v_2) = \left(\frac{F-k}{G} , G \right) = \left( b_2, 0 \right), \\  
p_3 : (f,g) &= (0,k), &&\quad p_7 : (u_3,v_3) = \left(\frac{f}{g-k} , g-k \right) = \left( b_3, 0 \right), \\  
p_4 : (F,G) &= (0, k), &&\quad p_8 : (u_4,v_4) = \left(\frac{F}{G-k} , G-k \right) = \left( b_4, 0 \right).
\end{align}
\end{subequations} 
Then $X_{\mathbf{b}}$ form a family of generic $D_4^{(1)}$-surfaces, with surface root basis in $\Pic(X)$ given by \eqref{hoffmanssurfaceroots} and symmetry root sublattice with basis \eqref{hoffmanssymmetryroots}. Further, we have a Cremona action of the symmetry group $\widetilde{W}(D_4^{(1)})$, which we write as an action on the parameters and coordinates defined as follows:
\begin{align*}
& r_0 :  \left\{ \begin{aligned}
f &\mapsto  \frac{(b_1+1)f g - k f - k g + k^2}{k f g - f +(b_1 - k^2) g + k}, \\
g &\mapsto  g,
\end{aligned}
\right. 
&& \begin{array}{cc}
r_0.b_1 = k^2-1 - b_1, & r_0.b_3 = \frac{1 - b_1 b_3}{k^2-1 -b_1}, \\
r_0.b_2 = \frac{(k^2-1)^2 - b_1 b_2}{k^2-1 -b_1}, & r_0.b_4 = \frac{1 - b_1 b_4}{k^2-1 -b_1},
\end{array}
\\
& r_1 : \left\{ \begin{aligned}
f &\mapsto  \frac{k f g + (b_2-k^2) f - k + k}{k^2 f g - k f - k g + (b_2+1)}, \\
g &\mapsto  g,
\end{aligned}
\right. 
&& \begin{array}{cc}
r_1.b_1 =  \frac{(k^2-1)^2 - b_1 b_2}{k^2-1 -b_2}, & r_1.b_3 = \frac{1 - b_2 b_3}{k^2-1 - b_2},\\
r_1.b_2 =  k^2-1 - b_2,    & r_1.b_4 =  \frac{1 - b_2 b_4}{k^2-1 - b_2},
\end{array}
\\
& r_2 : \left\{ \begin{aligned}
f &\mapsto  \frac{g-k}{k g -1}, \\
g &\mapsto  \frac{f-k}{k f - 1},
\end{aligned}
\right. 
&& \begin{array}{cc}
r_2.b_1 =  \frac{(k^2-1)^2}{b_1}, &  r_2.b_3 =  \frac{1}{b_3(k^2-1)^2},\\
r_2.b_2 =  \frac{(k^2-1)^2}{b_2},   & r_2.b_4 = \frac{1}{b_4(k^2-1)^2},
\end{array}
\\
& r_3 : \left\{ \begin{aligned}
f &\mapsto  \frac{f(k-g)\left( (k^2-1)b_3 - 1\right)}{f(k g - 1)- b_3 (k^2-1)(g-k)}, \\
g &\mapsto  g,
\end{aligned}
\right. 
&& \begin{array}{cc}
r_3.b_1 =  \frac{(k^2-1)(b_1 b_3-1)}{(k^2-1) b_3- 1}, & r_3,b_2 =  \frac{(k^2-1)(b_2 b_3-1)}{(k^2-1) b_3- 1},  \\
r_3.b_3 = (k^2-1)^{-1} - b_3,  & r_3.b_4 =  \frac{(k^2-1)^2 b_3 b_4 -1}{ (k^2-1)^2 b_3 - k^2+1}, 
\end{array}
\\
& r_4 : \left\{ \begin{aligned}
f &\mapsto  \frac{ b_4 (k^2-1) f( k g -1) - g + k}{(k g -1)( b_4(k^2-1) - 1)}, \\
g &\mapsto  g,
\end{aligned}
\right.
&& \begin{array}{cc}
r_4.b_1 =  \frac{(k^2-1)^2(1 - b_1 b_4)}{(k^2-1)b_4 -1}, & r_4.b_2 =  \frac{(k^2-1)^2(1 - b_2 b_4)}{(k^2-1)b_4 -1},  \\
 r_4.b_3 = \frac{(k^2-1)^2 b_3 b_4 - 1}{1-k^2 - b_4(k^2-1)^2}, &  r_4.b_4 = (k^2-1)^{-1} - b_4, 
\end{array}
\end{align*}

\begin{align*}
& \pi_{(01)}:  \left\{ \begin{aligned}
f &\mapsto  f, \\
g &\mapsto  1/g,
\end{aligned}
\right. 
&& \begin{array}{ccc}
\pi_{(01)}.k = 1 /k, & \pi_{(01)}.b_1 = -b_2 / k^2, & \pi_{(01)}.b_3 = -k^2 b_3, \\
 &        			 \pi_{(01)}.b_2 = -b_1 / k^2, & \pi_{(01)}.b_4 = -k^2 b_4,
\end{array}
\\
& \pi_{(34)}:  \left\{ \begin{aligned}
f &\mapsto  1/f, \\
g &\mapsto  g,
\end{aligned}
\right. 
&& \begin{array}{ccc}
\pi_{(34)}.k = \frac{1}{k}, & \pi_{(34)}.b_1 = -b_1/k^2,& \pi_{(34)}.b_3 = -k^2 b_4, \\
&                		 \pi_{(34)}.b_2 = - b_2 / k^2, & \pi_{(34)}.b_4 = -k^2 b_3,
\end{array}
\\
& \pi_{(14)}:  \left\{ \begin{aligned}
f &\mapsto  \frac{f (k^2-1)^{1/2}}{k f - 1}, \\
g &\mapsto  \frac{g (k^2-1)^{1/2}}{k g - 1},
\end{aligned}
\right. 
&& \begin{array}{ccc}
\pi_{(14)}.k = k /(k^2-1)^{1/2}, & \pi_{(14)}.b_1 = b_1/(k^2-1)^2, & \pi_{(14)}.b_2 = b_4, \\
 &               		 \pi_{(14)}.b_3 = (k^2 -1)^2 b_3, & \pi_{(14)}.b_4 =b_2.
\end{array}
\end{align*}
Through the pullbacks of these birational maps and the identification of the Picard groups, this representation corresponds to the action of $\widetilde{W}(D_4^{(1)})$ on $\Pic(X)$ by Cremona isometries.
\end{proposition}

\begin{proof}
The birational mapping defined by the action of each generator may be verified to give the required change of blowing-down structure and parameters by similar calculations to those involved in the proofs of Proposition 3.2 and Theorem 3.3. For helpful discussions of the methods for constructing such a Cremona action we refer the reader to, for example, one of \cite{SAKAI2001, MSY2003,DT2018,JNS2016,KNY2017}.
\end{proof}

Before we construct the generic version, we illustrate how to recover the system \eqref{hoffmans} from the action of $\varphi$ on the variables and parameters. To reconstruct the equation using this birational representation, we compute the action of the element $\varphi = r_{201234}$ on the variables and parameters. For composing the actions of different elements of the group, we note that $\widetilde{W}(D_4^{(1)})$ acts on rational functions of the variables and parameters by replacement, so for $w \in \widetilde{W}(D_4^{(1)})$ and rational function $R \in \C ( f, g ; \mathbf{b})$, we have
\begin{equation}
w. R(f, g ; \mathbf{b}) = R ( w.f, w.g ; w. \mathbf{b}).
\end{equation}
Using the formulae given in proposition 4.1, we may compute $\varphi.f$ and $\varphi.g$, which are complicated rational functions of the variables and parameters, which we omit for the moment for conciseness. Direct computation shows that $\varphi.\mathbf{b} = \left( \varphi.k, \varphi.b_1, \varphi.b_2, \varphi.b_3,\varphi.b_4 \right)$ is given by
\begin{subequations} \label{phib}
\begin{align}
\varphi.k &= k, \\
 \varphi.b_1 &= b_2 + \frac{b_2 b_3 - 1}{b_3 + t} + \frac{b_2 b_4 -1}{b_4 + t} + \frac{b_1 - b_2}{
 1 + b_1 t}, \\
 \varphi.b_2 &= b_1 + \frac{ b_1 b_3 -1}{b_3 + t} + \frac{b_1 b_4 -1}{b_4 + t} + \frac{b_2 - b_1}{
 1 + b_2 t}, \\
 \varphi.b_3 &= \frac{ (1 + b_1 t) (1 + b_2 t) (b_3 + t) b_4 }{(b_4 + t) \left(2 b_1 b_2 b_3 t + b_1 b_2 t^2 +(b_1 + b_2) b_3 -1 \right)}, \\
 \varphi.b_4 &= \frac{ (1 + b_1 t) (1 + b_2 t) (b_4 + t) b_3 }{(b_3 + t) \left(2 b_1 b_2 b_4 t + b_1 b_2 t^2 +(b_1 + b_2) b_4 -1 \right)},
\end{align}
\end{subequations}
where we have denoted $t = 1/ (1-k^2)$. We note that this is not a translational motion in the parameter space, but if we restrict to the case when $\mathbf{b}= (k, b_1, b_1, 0, 0)$, then $\varphi$ acts on this parameter subspace by translation, with 
\begin{equation}
\varphi.(k, b_1, b_1, 0, 0)= (k, b_1 - 2 / t, b_1- 2 / t, 0, 0).
\end{equation}
Therefore the result of iterating the action on parameters $n$ times is given by $\varphi^n .(k, b_1, b_1, 0, 0) = (k, b_1 - 2n/t , b_1 - 2n/t, 0,0 )$, so letting $f_n = \varphi^n.f, g_n = \varphi^n.g_n$, we obtain the following system of difference equations.:
\begin{subequations}
\begin{align}
\bar{f} &= \frac{ (\bar{g}-k)^2}{f ( k \bar{g} - 1)^2}, \\
\bar{g} &= \frac{(f-k)\left( k(1-k^2) f g + (b_1 + 2n(k^2-1) ) f + (k^2-1) g - k(b_1+ 2n(k^2-1)) \right)}{(kf-1)\left( (k(b_1 + 2n(k^2-1)) f g + (1-k^2)f - (b_1 + 2n(k^2-1))g + k (k^2-1) \right) }.
\end{align}
\end{subequations}
Setting the initial value of the parameter to be $b_1 = k^2-1 = -1/t$, we recover \eqref{hoffmans}, as well as the family of surfaces $X_{(k,n)}$ constructed as its space of initial values in section 3. 

\begin{remark}
Firstly, we note that the process by which we recovered the system \eqref{hoffmans} from the birational representation is a projective reduction, in the sense that the system is given by the Cremona action of a non-translation element of infinite order, with a restriction to a parameter subspace on which the action is translational. However, with the root variable parametrisation, the action of $\varphi$ becomes translational on a larger parameter subspace than the one defined by the constraint $\mathbf{b} = (k, b_1, b_1 , 0 ,0 )$. This differs from the example \cite{KNT2011} described in the introduction, in which the projection was onto the maximal parameter subspace on which the half-translation $R$ gave translational motion. This, together with the fact that \eqref{hoffmans} was recovered by a special choice of the initial value of the parameter $b_1$, leads us to interpret the system as a special case of a projective reduction of the discrete Painlev\'e equation associated with the translation $T_{2v}$. \\

\end{remark}

\subsection{Transformation to the canonical model and root variable parametrisation}
To write down a difference equation from the action of $\varphi$ on the variables and parameters in the generic case, we let $(f_n, g_n) = \varphi^{n}. (f,g)$ be the result of acting by $\varphi$ on the variales $n$ times, and seek an explicit form of $\varphi^{n} .\mathbf{b}$ as a function of $n$ with parameters $\mathbf{b} = (k, b_1,b_2,b_3,b_4)$. Obtaining this expression amounts to solving the system of difference equations \eqref{phib}, which are nonlinear in their present form. However, if we change our parametrisation of the family of surfaces to that given by the root variables, the transformation of parameters will linearise the system \eqref{phib} and it will be explicitly solvable by elementary methods. This will be the case for the action of any element of infinite order on the root variables, so difference equations can always be explicitly constructed from these elements, as we will demonstrate. \\

Rather than obtaining the root variable parametrisation from the period map by directly computing integrals, we will obtain a transformation to the canonical form of a family of generic $D_4^{(1)}$ surfaces with root variable parametrisation given by Kajiwara, Noumi and Yamada \cite{KNY2017}, which we call the KNY form. Incidentally, this is Okamoto's space of initial conditions for $\pain{VI}$, so reveals \eqref{hoffmans} as a B\"acklund transformation for a special case. Our method for constructing this transformation will be to first find a map between $\Pic(X)$ and the corresponding module for the KNY form of the family, which identifies the surface and symmetry roots, then obtain a birational map that induces it. This is similar to that employed in the recent work \cite{DT2018}, though here we emphasise its application in obtaining the root variables, direct computation of which is not practical in our case. \\

We begin by recalling the point configuration and root data for the KNY family of generic $D_4^{(1)}$-surfaces. Considering $\C^2$ with coordinates $(x,y)$ and compactifying to $\p^1 \times \p^1$, we obtain the family of surfaces $Y_{\mathbf{a}}$ with root variable parametrisation $\mathbf{a} = (t , a_0, a_1, a_2, a_3, a_4 )$, by blowing up eight points given in coordinates by
\begin{subequations}
\begin{align}
\tilde{p}_1 : (x,y) &= (\infty , - a_2), &&\quad \tilde{p}_2 : (x,y) = (\infty, - a_1 - a_2), \\  
\tilde{p}_3 : (x,y) &= (t, \infty) , &&\quad \tilde{p}_4 : (\tilde{u}_3,\tilde{v}_3) = \left(y(x-t) , \frac{1}{y} \right) = \left( t a_0, 0 \right), \\  
\tilde{p}_5 : (x,y) &= (0,0), &&\quad \tilde{p}_6 : (x,y) = \left( 0, a_4 \right), \\  
\tilde{p}_7 : (x,y) &= (1, \infty), &&\quad \tilde{p}_8 : (\tilde{u}_7,\tilde{v}_7) = \left(y(x-1) , \frac{1}{y} \right) = \left( a_3, 0 \right).
\end{align}
\end{subequations} 
Denoting divisor classes of total transforms of lines of constant $x, y$ by $\H_x, \H_y$ respectively and the exceptional classes arising from the blowups of points $\tilde{p}_1, \dots, \tilde{p}_8$ by $\L_1, \dots, \L_8$, we have
\begin{equation}
\Pic(Y) = \Z \H_x \oplus \Z \H_y \oplus \Z \L_1 \oplus \dots \oplus \Z \L_8,
\end{equation}
with the surface root basis $\{\tilde{\delta}_0, \dots , \tilde{\delta}_4\}$ and symmetry root basis $\left\{\tilde{\alpha}_0, \dots , \tilde{\alpha}_4\right\}$ given by 
\begin{subequations}
\begin{align}
\tilde{\delta}_0 &= \L_3 - \L_4, \quad \quad  &&\tilde{\alpha}_0 = \H_x - \L_3 - \L_4, \\
\tilde{\delta}_1 &= \H_x - \L_1 - \L_2, \quad \quad &&\tilde{\alpha}_1 = \L_1 - \L_2, \\
\tilde{\delta}_2 &= \H_y - \L_3 - \L_7, \quad \quad &&\tilde{\alpha}_2 = \H_y - \L_1 - \L_5, \\
\tilde{\delta}_3 &= \L_7 - \L_8, \quad \quad &&\tilde{\alpha}_3 = \H_x - \L_7 - \L_8, \\
\tilde{\delta}_4 &= \H_x - \L_5 - \L_6,\quad \quad  &&\tilde{\alpha}_4 = \L_5 - \L_6.
\end{align}
\end{subequations}
The following proposition, which provides an isomorphism of modules preserving the surface and symmetry root lattices associated with the two families of generic $D_4^{(1)}$-surfaces, may be proved by direct calculation.
\begin{proposition}
The linear map $\psi : \text{Pic}(X) \rightarrow \text{Pic}(Y)$ defined by
\begin{equation}
\psi : \left\{ \begin{array}{lcl}
\H_f & \mapsto & \H_x + \H_y -  \L_1-\L_5, \\
\H_g & \mapsto & \H_x,\\
\E_1 & \mapsto & \L_3, \\
\E_2 & \mapsto & \H_x - \L_1, \\
\E_3 & \mapsto & \L_7, \\
\E_4 & \mapsto & \H_x- \L_5, \\
\E_5 & \mapsto &  \L_4, \\
\E_6 & \mapsto &  \L_2, \\
\E_7 & \mapsto & \L_8, \\
\E_8 & \mapsto & L_6,
\end{array}
\right.
\end{equation}
preserves the intersection product and identifies the canonical classes in $\Pic(X), \Pic(Y)$. \\
Further, we have that 
\begin{equation}
\psi (\alpha_i) = \tilde{\alpha}_i, \quad \quad \psi(\delta_i) = \tilde{\delta}_i.
\end{equation}
\end{proposition}
We next demonstrate how to find a change of variables inducing this map, which will determine the relation between parameters $\mathbf{a}$ and $\mathbf{b}$.

\begin{proposition}
The map $\psi$ is induced by the pushforward $\Psi_*$ of 
\begin{equation}
\Psi : X_{\mathbf{b}} \rightarrow Y_{\mathbf{a}},
\end{equation}
defined by
\begin{equation} \label{fgtoxy}
x = \frac{k g- 1}{k^2-1}, \quad \quad y = \frac{a_2}{k} \frac{(1 - k f)( k g -1)}{(k f g - f - g + k )} ,
\end{equation}
with the correspondence between parameters $\mathbf{a}$ and $\mathbf{b}$ given by
\begin{equation} \label{parametercorrespondence}
t = \frac{1}{1-k^2}, \quad b_1 = -\frac{ a_0+a_2 }{ t a_0}, \quad b_2 = -\frac{a_1 + a_2}{t a_1}, \quad b_3 = -\frac{t( a_2+a_3)}{a_3}, \quad b_4 = -\frac{t( a_2+a_4)}{a_4}.
\end{equation}
\end{proposition}
\begin{proof}
We will obtain the transformation by lifting a map defined by 
\begin{equation}
(f,g) \mapsto \left( x, y \right), 
\end{equation}
where $x$ and $y$ are rational functions of $f,g$ and the parameters. The form of these rational functions is determined by the condition that $\Psi_* = \psi$. In particular, as $\Psi_*(\H_g) = \H_x$, the total transform in $X_{\mathbf{b}}$ of a line of constant $g$ must be sent under $\Psi$ to the total transform of a curve of constant $x$ in $Y_{\mathbf{a}}$. This implies that the rational function $x(f,g)$ is a fractional linear transformation of $g$ independent of $f$, and we obtain our ansatz for $x$:
\begin{equation} \label{ansatzx}
x = \frac{ \lambda_1 g + \lambda_2}{\lambda_3 g + \lambda_4},
\end{equation}
where $\lambda_1, \dots ,\lambda_4$ are complex constants to be determined. Further, as $\Psi_*(\H_f) = \H_x + \H_y - \L_1 - \L_5$, the total transform in $X_{\mathbf{b}}$ of a line of constant $f$ must be sent to the proper transform of a curve of bi-degree $(1,1)$ in $(x,y)$ passing through $\tilde{p}_1$ and $\tilde{p}_2$ each with multiplicity one. Taking into account the fact that $x$ and $g$ are related by M\"obius transformation, the desired correspondence dictates that $y$ be of the form
\begin{equation}\label{ansatzy}
y = \frac{\mu_1 f g + \mu_2 f + \mu_3 g + \mu_4}{\mu_5 f g + \mu_6 f + \mu_7 g + \mu_8},
\end{equation}
where $\mu_1, \dots, \mu_8$ are to be determined. \\
We now demonstrate how the unknown coefficients in \eqref{ansatzx} and \eqref{ansatzy} can be deduced from the requirement that $\Psi$ induces $\psi$ by pushforward. Similarly to the way in which the Cremona isometry in Theorem 3.3 was deduced from the birational maps by one of many possible sequences of calculations, here we also have the freedom to choose how we deduce the conditions on the coefficients, and present one which simplifies calculations.\\
We first note that the general forms \eqref{ansatzx} and \eqref{ansatzy} determine the coefficients of $\H_x, \H_y$ in the pushforwards of $\H_f, \H_g$, but not any of the exceptional classes. Given that 
\begin{equation}
\psi^{-1}(\H_y) = \H_f + \H_g - \E_2 - \E_4,
\end{equation}
we see that an arbitrary line of constant $y$ should be sent by $\Psi^{-1}$ to a $(1,1)$-curve in $(f,g)$ coordinates passing through $p_2$ and $p_4$, each with multiplicity one. Indeed, letting $c \in \C$ be arbitrary and setting $y = c$ in \eqref{ansatzy} we obtain a $(1,1)$-curve, which we rewrite in the chart $(F,G)$, in which $p_2, p_4$ are visible, as
\begin{equation}
c \left( \mu_5 + \mu_6 G + \mu_7 F + \mu_8 F G\right) = \mu_1 + \mu_2 G + \mu_3 F + \mu_4 F G.
\end{equation}
Requiring that this intersects $p_2 : (F, G) = (k,0)$, we obtain
\begin{equation}
c( \mu_5 + k \mu_7) = \mu_1 + k \mu_3.
\end{equation}
Since the curve must intersect $p_2$ independent of the value of the constant $c$, we require
\begin{equation}
\mu_1 = -k \mu_3, \quad \quad \mu_5 = - k \mu_7.
\end{equation}
Similarly, the requirement that any such curve should pass through $p_4 : (F,G) = (0,k)$ yields the conditions
\begin{equation}
\mu_6 = \mu_7, \quad \quad \mu_3 = \mu_2,
\end{equation}
after which our ansatz reads
\begin{equation} \label{refinedy}
y = \frac{ \mu_4 (k f g - f - g) + k \mu_1}{k (k f g - f -g + k )}.
\end{equation}
Before finding the coefficients $\mu_1, \mu_4$, it will be convenient to first determine the fractional linear transformation \eqref{ansatzx} relating $x$ and $g$. \\
We consider lines corresponding to special values of $x$ and $g$, whose proper transforms give unique representatives of their divisor classes. For example, as $\psi(\H_g - \E_3) = \H_x - \L_7$, the line $g=k$ should be sent under $\Psi$ to the line $x=1$. Computing this condition explicitly, we obtain
\begin{equation}
\lambda_4 = k (\lambda_1 - \lambda_3) + \lambda_2.
\end{equation}
Similarly, from the fact that $\psi(\H_g - \E_1) = \H_x - \L_3$, the lines $g=k$ and $x=t$ should be in correspondence, so we may obtain $\lambda_2$ in terms of the other coefficients and parameters, and our ansatz is refined:
\begin{equation} \label{refinedansatz}
x = \frac{ \lambda_1(t-1) g + k t (\lambda_3 - \lambda_1)}{ \lambda_3(t-1) g + k (\lambda_1 - \lambda_3)}.
\end{equation}
A third condition may be deduced from the fact that $\psi(\H_g - \E_2) = \L_1$, so the proper transform of the line $G= 0$ should be mapped to the exceptional divisor replacing $\tilde{p}_1$, and in particular setting $G=0$ in \eqref{refinedansatz} should recover its $x$-coordinate. Direct calculation yields $\lambda_3 = 0$, and therefore
\begin{equation} \label{fullx}
x = \frac{(1-t)g + k t}{k}.
\end{equation}
The parameters $k$ and $t$ each correspond to the `extra parameter' \cite{SAKAI2001} in their respective families of surfaces, and their relationship can be deduced as follows. We have already found the unknown coefficients in the relation between $g$ and $x$, by considering three special values of $g$ corresponding to the coordinates of the basepoints $p_1, p_2, p_3$. However, there is one more special line of constant $g$ which we have not considered, and we must verify that \eqref{fullx} is consistent with the required image of the corresponding divisor class under $\psi$. To be precise, the divisor class of the proper transform of the line $G=k$ is sent under $\psi$ according to $\psi(\H_g - \E_4) = \L_5$, so setting $G =k$ in \eqref{fullx} should recover the $x$-coordinate of $\tilde{p}_5$, namely $x=0$. This condition is computed to be equivalent to
\begin{equation}
t = \frac{1}{1-k^2},
\end{equation}
so we have obtained a necessary correspondence between the parameters $k$ and $t$ such that the birational map $\Psi : X_{\mathbf{b}} \rightarrow Y_{\mathbf{a}}$ induces $\psi$. \\
We now return to our expression \eqref{refinedy}, and use similar calculations to determine the remaining unknown coefficients. For instance, in order to satisfy the condition $\Psi_*(\H_g - \E_2) = \L_1$, we require that setting $G = 0$ with $F \neq k$ in \eqref{refinedy} and \eqref{fullx} gives the $(x,y)$ coordinates of $\tilde{p}_1$. We have already imposed this condition on the expression for $x$, so substitute in to \eqref{refinedy} to obtain 
\begin{equation}
\mu_4 = - k a_2.
\end{equation}
A similar calculation based on $\psi(\H_g - \E_4) = \L_5$ leads to the condition 
\begin{equation}
\mu_1 = -(a_1 + a_2)/k,
\end{equation}
and we have determined the birational map completely and obtained \eqref{fgtoxy}. \\
We now demonstrate how to obtain the correspondences between the remaining parameters. We require that $\Psi_*(\E_5) = \L_4$, so the basepoint $p_5$ on the exceptional line $E_1$ in $X_{\mathbf{b}}$ should be mapped under $\Psi$ to $\tilde{p}_4$ on the line $L_3$ in $Y_{\mathbf{a}}$. To compute this condition, we use the previously defined charts $(u_1, v_1)$ on the exceptional line $E_1$, and $(\tilde{u}_3, \tilde{u}_3)$ for $L_3$ given by
\begin{equation}
\tilde{u}_3 = y(x-t), \quad \quad \tilde{v}_3 = \frac{1}{y}.
\end{equation}
Substituting these coordinates in \eqref{fgtoxy}, we obtain 
\begin{equation} \label{correspondenceE1L3}
\tilde{u}_3 = \frac{a_2(k v_1 -1)(1- k^2 - k u_1 v_1)}{(k^2-1)(k^2-1 - u_1 + k u_1 v_1)}, \quad \quad \tilde{v}_3 = \frac{k v_1(k^2-1 - u_1 + k u_1 v_1)}{a_2(k v_1 - 1)(1- k^2 - k u_1 v_1)},
\end{equation}
in which setting $v_1 = 0$ leads to
\begin{equation}
\tilde{u}_3 = \frac{a_2}{k^2 - 1 - u_1}, \quad \quad \tilde{v}_3 = 0,
\end{equation}
which demonstrates that $\Psi$ indeed gives a one-to-one correspondence between the exceptional lines $E_1$ and $L_3$. Requiring that $p_5 : (u_1, v_1) = (b_1, 0)$ is mapped to $\tilde{p}_4 : (\tilde{u}_3, \tilde{v}_3 ) = (t a_0, 0 )$, we obtain the condition
\begin{equation}
b_1 = - \frac{ a_0 + a_2}{t  a_0 }.
\end{equation}
Similar calculations based on the requirements $\psi(\E_7) = \L_8, ~\psi(\E_6) = \L_2,~ \psi(\E_8) = \L_6$ yield the rest of the correspondences \eqref{parametercorrespondence}. The proof is completed by checking, using the same calculation methods as in previous parts, that the birational map obtained does indeed induce $\psi$ as its pushforward when parameters are identified according to \eqref{parametercorrespondence}.
\end{proof}
We note that this transformation, and in particular the parameter correspondence, recovers the root variable parametrisation of the family $X_{\mathbf{b}}$ associated with the simple root basis $\left\{ \alpha_0, \dots, \alpha_4 \right\} \subset \Pic(X)$. The map $\Psi$ gives a change in blowing-down structure of the $D_4^{(1)}$-surface $X = X_{\mathbf{b}}$, under which the period map is invariant, in the sense that $\chi_{X} (\lambda) = \chi_{\Psi(X)} (\Psi_{*}(\lambda))$ for $\lambda \in \Pic(X)$, as has been noted in \cite{DT2018}. Therefore we may obtain the root variables of the surface $X = X_{\mathbf{b}}$ as
\begin{equation}
\chi_{X} (\alpha_j) = \chi_{Y}(\Psi_*(\alpha_j)) = \chi_Y (\tilde{\alpha}_j) = a_j.
\end{equation}
In particular, the correspondence \eqref{parametercorrespondence} gives the root variable parametrisation of the family $X_{\mathbf{b}}$, which we write as $X_{\mathbf{a}}$ with $\mathbf{a} = (t, a_0, a_1, a_2, a_3, a_4)$, with the normalisation $a_0 + a_1 +2a_2 + a_3 + a_4 = 1$.

\begin{remark}
Solving the system \eqref{parametercorrespondence} for the root variables $a_0,\dots, a_4$, we see that they are given by rational functions of degree three in the parameters $b_1, \dots, b_4$, which suggests that if we were to compute them directly using the period mapping it would have involved evaluating complicated integrals, which this transformation has allowed us to avoid.
\end{remark}
%
We also note that the family of surfaces $Y_{\mathbf{a}}$ forms the space of initial values for $\pain{VI}$, with the coordinates $(x,y)$ related to the variables in the Hamiltonian form \cite{KNY2017} according to 
\begin{equation} \label{P6qp}
q = x, \quad \quad p = \frac{y}{x},
\end{equation}
with $t$ playing the role of the independent variable. In particular, the birational representation of $\widetilde{W}(D_4^{(1)})$ given in Proposition 4.1 corresponds to the B\"acklund transformation symmetries of $\pain{VI}$ in the sense that for each $w\in \widetilde{W}(D_4^{(1)})$, we have a birational map
\begin{equation}
w : X_{\mathbf{a}} \rightarrow X_{w.\mathbf{a}},
\end{equation}
through the action on the variables $(f,g)$ and root variables parameters, which we may conjugate by our transformation to obtain a birational map
\begin{equation}
\tilde{w} : Y_{\mathbf{a}} \rightarrow Y_{w.\mathbf{a}}, 
\end{equation}
which gives a B\"acklund transformation on the variables $(q,p)$ via the relation \eqref{P6qp}. We remark that it may be checked that conjugating the actions from Proposition 4.1 by the transformation $\Psi$ in this way recovers the birational representation of $\widetilde{W}(D_4^{(1)})$ on the family of surfaces $Y_{\mathbf{a}}$ given in \cite{KNY2017}.

\subsection{Full-parameter system}

We are now in a position to construct the full-parameter version of \eqref{hoffmans} using our Cremona action of $\widetilde{W}(D_4^{(1)})$ on the family of generic $D_4^{(1)}$-surfaces $X_{\mathbf{a}}$ with the root variable parametrisation. We emphasise again that the key to writing down an explicit form of the discrete system defined by the Cremona action of an element of the symmetry group is to obtain a closed form for the result of acting on the parameters $n$ times. In our case, we are able to do this as the system \eqref{phib} giving the action of the element $\varphi = r_{201234}$ on the parameters is linearised by the transformation to root variables, and we now have 
\begin{equation}
\varphi . (t, a_0, a_1, a_2, a_3, a_4) = (t, a_1, a_0, a_2+1, a_4 - 1, a_3 - 1).
\end{equation}
This leads to a linear difference equation for $\varphi^n.\mathbf{a}$, which we may solve explicitly by elementary methods to arrive at the following formulae:
\begin{subequations}
\begin{align}
\varphi^n.t &= t, \quad \quad \varphi^n.k = k, \\
\varphi^n. a_0 &= \frac{1}{2} \left( (1 + (-1)^n)a_0 + ( 1 - (-1)^n)a_1 \right) = \begin{cases} a_0 & \text{for } n \text{ even}, \\ a_1 & \text{for } n \text{ odd}, \end{cases} \\
\varphi^n. a_1 &= \frac{1}{2} \left( (1 - (-1)^n)a_0 + ( 1 + (-1)^n)a_1 \right) = \begin{cases} a_1 & \text{for } n \text{ even}, \\ a_0 & \text{for } n \text{ odd}, \end{cases} \\
\varphi^n. a_2 &= a_2 + n, \\
\varphi^n. a_3 &= \frac{1}{2} \left( (1 + (-1)^n)a_3 + ( 1 - (-1)^n)a_4 - 2n \right) = \begin{cases} a_3 - n & \text{for } n \text{ even}, \\ a_4 - n & \text{for } n \text{ odd}, \end{cases} \\
\varphi^n. a_4 &= \frac{1}{2} \left( (1 - (-1)^n)a_3 + ( 1 + (-1)^n)a_4 - 2n  \right) = \begin{cases} a_4 - n & \text{for } n \text{ even}, \\ a_3 - n & \text{for } n \text{ odd}. \end{cases}
\end{align}
\end{subequations}
Therefore by setting $f_n = \varphi^n .f$ and $g_n = \varphi^n .g$, we obtain the explicit form of our full-parameter version of equation \eqref{hoffmans}, which we write with $(f,g) = (f_n,g_n)$ and $(\bar{f},\bar{g}) = (f_{n+1}, g_{n+1})$ as follows:
\begin{subequations} \label{fullgenericequation}
\begin{align}
\bar{f} &= \frac{ (\bar{g} - k)}{(k \bar{g} -1)} \frac{ \left( (a_2 +  a_{34}^{(n+1)}) f ( k\bar{g}-1 ) + (\bar{g} - k)(a_{34}^{(n)} - n-1)\right)}{\left( (a_2 + a_{34}^{(n)})(\bar{g} - k) + f (k \bar{g}-1)(a_{34}^{(n+1)} - n - 1) \right)}, \\
\bar{g} &= \frac{(f - k)}{(kf - 1)} \frac{ \left( (f-k) (a_2 + n ) - a_{01}^{(n+1)}(k f g - f - g + k)\right)}{\left( g (k f - 1) (a_2 + n) + a_{01}^{(n)}( k f g - f -g +k)  \right)},
\end{align}
\end{subequations}
\begin{equation*}
\text{where} \quad \quad \quad a_{01}^{(n)} = \begin{cases} a_0 & \text{for } n \text{ even}, \\ a_1 & \text{for } n \text{ odd}, \end{cases} \quad \quad a_{34}^{(n)} = \begin{cases} a_3 & \text{for } n \text{ even}, \\ a_4 & \text{for } n \text{ odd}, \end{cases}
\end{equation*}
and $a_0 +a_1+2a_2 +a_3 + a_4 =1$. The parameter specialisation $(b_1, b_2, b_3, b_4) = (k^2-1, k^2-1, 0, 0)$ by which we recovered \eqref{hoffmans} in subsection 4.1 is given in terms of the root variables as $(a_0,a_1,a_2, a_3, a_4) = (1/2, 1/2, 0 ,0,0)$, substitution of which in \eqref{fullgenericequation} again recovers the original system. 

\section{Integrability and other examples from non-translation symmetries}

We now discuss the integrability of the full-parameter system \eqref{fullgenericequation}, as well as how similar systems may be obtained for other surface types in Sakai's list. Firstly, we note that system \eqref{fullgenericequation} has algebraic entropy zero, which follows from the same result of Takenawa \cite{TAKENAWA2000} guaranteeing this property for the discrete Painlev\'e equations arising from translations. If we consider the iterates $(f_n,g_n)$ of a discrete system given by the Cremona action of some element of infinite order in the extended affine Weyl group as rational functions of arbitrary initial conditions $(f,g)$, their degrees may be computed using the intersection product on $\Pic(X)$ by the method of Takenawa, which we now recall following \cite{TAKENAWA2000}.\\
Consider a rational map given in inhomogeneous coordinates $(f,g)$ for $\p^1 \times \p^1$ by
\begin{equation}
\begin{aligned}
\eta : &~~\p^1 \times \p^1 \rightarrow \p^1 \times \p^1, \\
&(f,g) \mapsto \left( P(f,g), Q(f,g) \right),
\end{aligned}
\end{equation}
and define its degree to be 
\begin{equation}
\deg(\eta) = \max \left\{ \deg P(f,g) , \deg Q(f,g) \right\},
\end{equation}
where the degree of a rational function $P(f,g)$ is defined as the maximum of the degrees of its numerator and denominator as bivariate polynomials. Letting the maps induced by iteration of the system \eqref{fullgenericequation} be $\eta_{n+1} : (f_n,g_n) \mapsto \left( f_{n+1}, g_{n+1} \right)$, we write the map giving the $n$th iterate as a function of arbitrary initial conditions $(f_0,g_0) = (f,g)$ as
\begin{equation}
\begin{aligned}
&\eta^{(n)} : ~~\p^1 \times \p^1 \rightarrow \p^1 \times \p^1,\\
\eta^{(n)} = \eta_{n} \circ \eta_{n-1} \circ &\dots \eta_{1} \circ \eta_0 : ~ (f,g) \mapsto \left( P_{n}(f,g), Q_{n}(f,g) \right),
\end{aligned}
\end{equation}
The sequence of degrees of iterates of the system is then given by $d_n = \deg \eta^{(n)}$, and its algebraic entropy \cite{BV1999} is defined as
\begin{equation}
\varepsilon = \lim_{n \rightarrow \infty} \frac{1}{n} \log d_n,
\end{equation}
in the cases where the limit exists. Such a discrete system is said to be integrable in the sense of vanishing algebraic entropy if the growth of the degrees is polynomial, or in other words if $\varepsilon = 0$.\\
Computing the first few iterates of the system \eqref{fullgenericequation}, we obtain the following:
\begin{subequations}
\begin{align}
\deg P_n(f,g) &= 1, 7, 21, 43, 73,  \dots \\
\deg Q_n(f,g) &= 1, 3, 13, 31, 57, \dots
\end{align}
\end{subequations}
This appears to be quadratic, which we may confirm by computing explicit forms for $\deg P_n$ and $\deg Q_n$, using the fact that the system lifts to a family of birational isomorphisms by construction. Letting this family be
\begin{equation}
\Phi_{n+1} : X_{\varphi^n.\mathbf{a}} \rightarrow X_{\varphi^{n+1}.\mathbf{a}},
\end{equation}
we obtain the maps $\eta_n$ through the blowup projections such that the following diagram commutes:
\begin{center}
\begin{tikzcd}[column sep=large]
X_{\varphi^n.\mathbf{a}} \arrow[r, "\Phi_{n+1}"] \arrow[d] & X_{\varphi^{n+1}.\mathbf{a}} \arrow[d] \\
\p^1 \times \p^1 \arrow[r , "\eta_{n+1}"] & \p^1 \times \p^1
\end{tikzcd}
\end{center}
The Cremona isometry induced by the pullbacks $\Phi_{n+1}^*$ is again denoted $\varphi$, and coincides with that given in theorem 3.3. We have the following useful formulae,
\begin{equation}
\deg P_n(f,g) =\varphi^{-n}( \H_f + \H_g) \cdot \H_f, \quad \quad \deg Q_n(f,g) =\varphi^{-n}( \H_f + \H_g) \cdot \H_g,
\end{equation}
in which we note that the inverse of $\varphi$ appears because we have defined $\varphi$ in terms of pullbacks, whereas these formulae in their original form \cite{TAKENAWA2000} considered the pushforwards. \\
Indeed, the linear map $\varphi^{-1} : \Pic(X) \rightarrow \Pic(X)$ gives a system of difference equations for the coefficients of $\H_f, \H_g, \E_1, \dots ,\E_8$ in $\varphi^{-n}(\H_f + \H_g)$, which we can solve and compute intersection numbers to obtain the degrees exactly as
\begin{equation}
\deg P_n(f,g) = 4 n^2 +2n +1, \quad \quad \deg Q_n(f,g) = 4 n^2 -2n +1.
\end{equation}
With these formulae, we have proven that the system \eqref{fullgenericequation} is integrable in the sense of vanishing algebraic entropy. In fact, any system obtained as the Cremona action of an element of the symmetry group of a family of $\mathcal{R}$-surfaces will have at most quadratic degree growth, which follows from the fact that the matrix giving the action on $\Pic(X)$ is a product of finite order matrices and therefore must have eigenvalues only on the unit circle in the complex plane.\\

We will now demonstrate how other systems may be obtained in the same way from non-translation elements of infinite order, through the birational representation of $\widetilde{W}( (A_2+A_1)^{(1)})$ discussed in the introduction. Consider again the element $R = \pi^2 s_1$, which we recall acts on the parameters according to 
\begin{equation}
R.(a_0, a_1, a_2, c) = (a_2 a_0, q^{-1} a_2 a_1 , q a_2^{-1} , c).
\end{equation}
Again, this gives a linear system of equations, this time of $q$-difference type, for $R^{n}.(a_0,a_1,a_2,c)$, which we may solve explicitly to obtain
\begin{subequations}
\begin{align}
R^n .a_0 &= q^{\frac{1}{4}(3 + (-1)^{n}+2n)} a_1^{-1} a_2^{-\frac{1}{2}(1 + (-1)^{n})} = \begin{cases} a_0 q^{n/2}  & \text{for } n \text{ even}, \\  a_1 q^{1/2+n/2} & \text{for } n \text{ odd}, \end{cases} \\
R^n .a_1 &= q^{\frac{1}{4}(-1 + (-1)^n- 2n)} a_1 a_2^{\frac{1}{2}(1+(-1)^{n})}= \begin{cases} a_1 q^{-n/2}  & \text{for } n \text{ even}, \\  a_1 a_2 q^{-1/2 -n/2} & \text{for } n \text{ odd}, \end{cases} \\
R^n .a_2 &= q^{\frac{1}{2}(1+ (-1)^{n+1})} a_2^{(-1)^n} = \begin{cases} a_2 & \text{for } n \text{ even}, \\  q a_2^{-1}& \text{for } n \text{ odd}, \end{cases}  \quad \quad R^n .c = c.
\end{align}
\end{subequations}
With this, we can write down an explicit form of the discrete system induced by the action of $R$ on the variables by setting $F_n = R^n.f_0$, and therefore obtain the following generic version of \eqref{qP2}:
\begin{equation} \label{genericR}
F_{n+1}F_{n-1}=\frac{q c^2}{ F_{n}} \frac{(a_0^{\frac{1}{2} \left((-1)^n-1\right)} + a_0 F_{n} q^{\frac{n}{2}} q^{\frac{1}{4} \left((-1)^n-1\right)})}{(a_0 q^{\frac{n}{2}} q^{\frac{1}{4} \left((-1)^n-1\right)}+a_0^{\frac{1}{2} \left((-1)^n-1\right)} F_{n} )}.
\end{equation}
Numerical experiments reveal the sequence of iterates of this system to have degree growth identical to that of the projective reduction \eqref{qP2}, which is to be expected given that they correspond to the same Cremona isometry. We now consider an element of $\widetilde{W}( (A_2+A_1)^{(1)})$ which may be regarded as a `twisted translation' similarly to the element $\varphi \in \widetilde{W}(D_4^{(1)})$ associated with the systems \eqref{hoffmans} and \eqref{fullgenericequation}, and construct another integrable full-parameter system from its Cremona action. Keeping the notation of \cite{KNT2011}, we consider the translation $T_4 = \sigma^3 w_0$, whose action on the parameters is given by 
\begin{equation}
T_4.(a_0, a_1, a_2, c) = (a_0, a_1, a_2 , q c). 
\end{equation}
It can be shown that the Dynkin diagram automorphism $\pi = \sigma^2$ preserves the Kac translation vector corresponding to $T_4$, and the element $\tilde{T}_4 = \pi T_4$ is of infinite order, and satisfies $\tilde{T}_4^3 = T_4^3$. Its action on the parameters is given by
\begin{equation}
\tilde{T}_4.(a_0, a_1, a_2, c) = (a_1, a_2, a_0, q c),
\end{equation}
from which we may deduce, by solving a linear system of difference equations as in the previous example, that the result of acting on the parameters $n$ times has the explicit form
\begin{equation}
\tilde{T}_4^n .c = q^n c, \quad \quad \tilde{T}_4^n .(a_0, a_1, a_2) = \begin{cases} (a_0,a_1,a_2) & \text{for } n \equiv 0 \mod 3, \\  (a_1,a_2,a_0) & \text{for } n \equiv 1 \mod 3, \\  (a_2,a_0,a_1) & \text{for } n \equiv 2 \mod 3.
\end{cases}  
\end{equation}
From this, we let $(F_n, G_n, H_n) = \tilde{T}_4^n. (f_0, f_1, f_2)$ and obtain the system of difference equations
\begin{subequations} \label{genericT4}
\begin{align}
F_{n+1} &= a_{(n+1)}a_{(n-1)} H_n \frac{1 + a_{(n)} F_n(1+a_{(n+1)} G_n)}{1 + a_{(n+1)} G_n (1+a_{(n-1)} H_n)}, \\
G_{n+1} &= a_{(n-1)}a_{(n)} F_n \frac{1 + a_{(n+1)} G_n(1+a_{(n-1)} H_n)}{1 + a_{(n-1)} H_n (1+a_{(n)} F_n)},\\
H_{n+1} &= a_{(n)}a_{(n+1)} G_n \frac{1 + a_{(n-1)} H_n(1+a_{(n)} F_n)}{1 + a_{(n)} F_n (1+a_{(n+1)} G_n)},\\
&F_n G_n H_n = q^{2n+1} c^2, \\
\text{where} &\quad \quad a_{(m)} = \begin{cases} a_0 & \text{for } m \equiv 0 \mod 3, \\  a_1 & \text{for } m \equiv 1 \mod 3, \\  a_2 & \text{for } m \equiv 2 \mod 3.
\end{cases}  
\end{align}
\end{subequations}
The method we have outlined for obtaining full-parameter systems from elements of infinite order works for any of the surface types in Sakai's list which admit non-translation elements of infinite order in their symmetry groups. This is due to the fact that the action on the root variables induces a linear system which may always be solved to obtain an explicit form of the result of acting on the parameters $n$ times. The equations \eqref{genericR}, \eqref{genericT4}, and more generally any systen given by the Cremona action on a family of generic $\mathcal{R}$-surfaces of a non-translation element will contain the maximum number of free parameters for their surface type. Moreover, they will be integrable in the sense of vanishing algebraic entropy for the same reason that Sakai's discrete Painlev\'e equations are. We note that there exist other systems, known as strongly asymmetric discrete Painlev\'e equations \cite{RGSTT2014, RGSTT2016}, in which the independent variable enters in a similar way to our generic non-translation equations. These were obtained through deautonomisation of QRT mappings, and it is natural to ask whether they arise as Cremona actions of non-translation elements on a family of generic $\mathcal{R}$-surfaces. In particular, we suspect that the periodic dependence of the coefficients on the independent variable may correspond to the power of the element of the symmetry group that gives a translation.

\section{Concluding remarks}
To summarise, we hope that this paper may be regarded as an effort to better understand the range of integrable systems arising from Sakai's framework, prompted by the geometric treatment of equation \eqref{hoffmans}. Through the geometric approach this system was found to be regularised on a family of $D_4^{(1)}$-surfaces and revealed as a special kind of projective reduction. In the process, a change of variables was obtained relating the system to $\pain{VI}$ which, though not clear by inspection of the equations, was constructed naturally using geometric methods. In particular, the parameter $k$ in system \eqref{hoffmans} was found to correspond to the independent variable in $\pain{VI}$, and an interesting question is whether known results on $\pain{VI}$ could be applied to discrete Amsler surfaces using this relationship. \\
We then demonstrated how generic (full-parameter) discrete systems can be constructed explicitly from the Cremona action on a family of generic $\mathcal{R}$-surfaces using any element of infinite order in the symmetry group. The method relied on the fact that with the root variable parametrisation, the action of the symmetry group is linear in either the additive, multiplicative or elliptic sense, so obtaining an explicit form for the result of acting on the parameters $n$ times is possible through solving a linear difference equation. We then demonstrated that any such system must have vanishing algebraic entropy, for essentially the same reason that the discrete Painlev\'e equations defined by translations do. Though previous studies of projective reductions have demonstrated that the range of integrable systems arising in Sakai's framework is not limited to those defined by translations, to our knowledge this is the first time that systems with the maximal number of parameters for their surface type have been constructed from non-translation symmetries.\\

\noindent \textbf{Acknowledgement:} The author would like to express his sincere thanks to Prof. R. Halburd for fruitful discussions and advice. This research was supported by a University College London Graduate Research Scholarship and Overseas Research Scholarship.

\bibliographystyle{abbrv}
\bibliography{paper1bibliography}

\def\cprime{$'$}
\begin{thebibliography}{10}

\bibitem{ADLER1998}
V.~E. Adler.
\newblock B\"acklund transformation for the {K}richever-{N}ovikov equation.
\newblock {\em Internat. Math. Res. Notices}, (1):1--4, 1998.

\bibitem{ABS2003}
V.~E. Adler, A.~I. Bobenko, and Y.~B. Suris.
\newblock Classification of integrable equations on quad-graphs. {T}he
  consistency approach.
\newblock {\em Comm. Math. Phys.}, 233(3):513--543, 2003.

\bibitem{AS2004}
V.~E. Adler and Y.~B. Suris.
\newblock {${\rm Q}\sb 4$}: integrable master equation related to an elliptic
  curve.
\newblock {\em Int. Math. Res. Not.}, (47):2523--2553, 2004.

\bibitem{AHJN2016}
J.~Atkinson, P.~Howes, N.~Joshi, and N.~Nakazono.
\newblock Geometry of an elliptic difference equation related to {Q}4.
\newblock {\em J. Lond. Math. Soc. (2)}, 93(3):763--784, 2016.

\bibitem{BV1999}
M.~P. Bellon and C.-M. Viallet.
\newblock Algebraic entropy.
\newblock {\em Comm. Math. Phys.}, 204(2):425--437, 1999.

\bibitem{DOLGACHEV1983}
I.~V. Dolgachev.
\newblock Weyl groups and {C}remona transformations.
\newblock In {\em Singularities, {P}art 1 ({A}rcata, {C}alif., 1981)},
  volume~40 of {\em Proc. Sympos. Pure Math.}, pages 283--294. Amer. Math.
  Soc., Providence, RI, 1983.

\bibitem{DT2018}
A.~Dzhamay and T.~Takenawa.
\newblock On some applications of sakai's geometric theory of discrete painlevŽ
  equations, April 2018.
\newblock preprint. arXiv:1804.10341.

\bibitem{RPG1991}
B.~Grammaticos, A.~Ramani, and V.~Papageorgiou.
\newblock Do integrable mappings have the {P}ainlev\'e property?
\newblock {\em Phys. Rev. Lett.}, 67(14):1825--1828, 1991.

\bibitem{RGSTT2014}
B.~Grammaticos, A.~Ramani, K.~M. Tamizhmani, T.~Tamizhmani, and J.~Satsuma.
\newblock Strongly asymmetric discrete {P}ainlev\'e equations: the additive
  case.
\newblock {\em J. Math. Phys.}, 55(5):053503, 17, 2014.

\bibitem{RGSTT2016}
B.~Grammaticos, A.~Ramani, K.~M. Tamizhmani, T.~Tamizhmani, and J.~Satsuma.
\newblock Strongly asymmetric discrete {P}ainlev\'e equations: the
  multiplicative case.
\newblock {\em J. Math. Phys.}, 57(4):043506, 30, 2016.

\bibitem{HHNS2015}
M.~Hay, P.~Howes, N.~Nakazono, and Y.~Shi.
\newblock A systematic approach to reductions of type-{Q} {ABS} equations.
\newblock {\em J. Phys. A}, 48(9):095201, 24, 2015.

\bibitem{HOFFMAN1999}
T.~Hoffmann.
\newblock Discrete amsler surfaces and a discrete painlev{\'e} iii equation.
\newblock {\em Oxford Lecture Series in Mathematics and its Applications},
  16:83--96, 1999.

\bibitem{JN2017}
N.~Joshi and N.~Nakazono.
\newblock Elliptic {P}ainlev\'e equations from next-nearest-neighbor
  translations on the {$E^{(1)}_8$} lattice.
\newblock {\em J. Phys. A}, 50(30):305205, 17, 2017.

\bibitem{JNS2016}
N.~Joshi, N.~Nakazono, and Y.~Shi.
\newblock Reflection groups and discrete integrable systems.
\newblock {\em Journal of Integrable Systems}, 1(1), 2016.

\bibitem{KAC}
V.~G. Kac.
\newblock {\em Infinite-dimensional Lie algebras}, volume~44.
\newblock Cambridge university press, 1994.

\bibitem{KNT2011}
K.~Kajiwara, N.~Nakazono, and T.~Tsuda.
\newblock Projective reduction of the discrete {P}ainlev\'e system of type
  {$(A_2+A_1)^{(1)}$}.
\newblock {\em Int. Math. Res. Not. IMRN}, (4):930--966, 2011.

\bibitem{KNY2017}
K.~Kajiwara, M.~Noumi, and Y.~Yamada.
\newblock Geometric aspects of {P}ainlev\'e equations.
\newblock {\em J. Phys. A}, 50(7):073001, 164, 2017.

\bibitem{RGKT2000}
M.~D. Kruskal, K.~M. Tamizhmani, B.~Grammaticos, and A.~Ramani.
\newblock Asymmetric discrete {P}ainlev\'e equations.
\newblock {\em Regul. Chaotic Dyn.}, 5(3):273--280, 2000.

\bibitem{LOOIJENGA1981}
E.~Looijenga.
\newblock Rational surfaces with an anticanonical cycle.
\newblock {\em Ann. of Math. (2)}, 114(2):267--322, 1981.

\bibitem{MMT1999}
T.~Matano, A.~Matumiya, and K.~Takano.
\newblock On some {H}amiltonian structures of {P}ainlev\'e systems. {II}.
\newblock {\em J. Math. Soc. Japan}, 51(4):843--866, 1999.

\bibitem{MATSUMIYA1997}
A.~Matumiya.
\newblock On some {H}amiltonian structures of {P}ainlev\'e systems. {III}.
\newblock {\em Kumamoto J. Math.}, 10:45--73, 1997.

\bibitem{MSY2003}
M.~Murata, H.~Sakai, and J.~Yoneda.
\newblock Riccati solutions of discrete {P}ainlev\'e equations with {W}eyl
  group symmetry of type {$E\sp {(1)}\sb 8$}.
\newblock {\em J. Math. Phys.}, 44(3):1396--1414, 2003.

\bibitem{RGO2001}
Y.~Ohta, A.~Ramani, and B.~Grammaticos.
\newblock An affine {W}eyl group approach to the eight-parameter discrete
  {P}ainlev\'e equation.
\newblock {\em J. Phys. A}, 34(48):10523--10532, 2001.
\newblock Symmetries and integrability of difference equations (Tokyo, 2000).

\bibitem{OKAMOTO1979}
K.~Okamoto.
\newblock Sur les feuilletages associ\'es aux \'equations du second ordre \`a
  points critiques fixes de {P}. {P}ainlev\'e.
\newblock {\em Japan. J. Math. (N.S.)}, 5(1):1--79, 1979.

\bibitem{QRT1988}
G.~R.~W. Quispel, J.~A.~G. Roberts, and C.~J. Thompson.
\newblock Integrable mappings and soliton equations.
\newblock {\em Phys. Lett. A}, 126(7):419--421, 1988.

\bibitem{QRT1989}
G.~R.~W. Quispel, J.~A.~G. Roberts, and C.~J. Thompson.
\newblock Integrable mappings and soliton equations. {II}.
\newblock {\em Phys. D}, 34(1-2):183--192, 1989.

\bibitem{RCG2009}
A.~Ramani, A.~S. Carstea, and B.~Grammaticos.
\newblock On the non-autonomous form of the {$Q\sb 4$} mapping and its relation
  to elliptic {P}ainlev\'e equations.
\newblock {\em J. Phys. A}, 42(32):322003, 8, 2009.

\bibitem{RG1996}
A.~Ramani and B.~Grammaticos.
\newblock Discrete {P}ainlev\'e equations: coalescences, limits and
  degeneracies.
\newblock {\em Phys. A}, 228(1-4):160--171, 1996.

\bibitem{RG2009}
A.~Ramani and B.~Grammaticos.
\newblock The number of discrete {P}ainlev\'e equations is infinite.
\newblock {\em Phys. Lett. A}, 373(34):3028--3031, 2009.

\bibitem{RGH1991}
A.~Ramani, B.~Grammaticos, and J.~Hietarinta.
\newblock Discrete versions of the {P}ainlev\'e equations.
\newblock {\em Phys. Rev. Lett.}, 67(14):1829--1832, 1991.

\bibitem{ST2002}
M.-H. Saito and T.~Takebe.
\newblock Classification of {O}kamoto-{P}ainlev\'e pairs.
\newblock {\em Kobe J. Math.}, 19(1-2):21--50, 2002.

\bibitem{STT2002}
M.-H. Saito, T.~Takebe, and H.~Terajima.
\newblock Deformation of {O}kamoto-{P}ainlev\'e pairs and {P}ainlev\'e
  equations.
\newblock {\em J. Algebraic Geom.}, 11(2):311--362, 2002.

\bibitem{SAKAI2001}
H.~Sakai.
\newblock Rational surfaces associated with affine root systems and geometry of
  the {P}ainlev\'e equations.
\newblock {\em Comm. Math. Phys.}, 220(1):165--229, 2001.

\bibitem{TS1997}
T.~Shioda and K.~Takano.
\newblock On some {H}amiltonian structures of {P}ainlev\'e systems. {I}.
\newblock {\em Funkcial. Ekvac.}, 40(2):271--291, 1997.

\bibitem{TAKENAWA2000}
T.~Takenawa.
\newblock Algebraic entropy and the space of initial values for discrete
  dynamical systems.
\newblock {\em J. Phys. A}, 34(48):10533--10545, 2001.
\newblock Symmetries and integrability of difference equations (Tokyo, 2000).

\bibitem{TAKENAWA2003}
T.~Takenawa.
\newblock Weyl group symmetry of type {$D^{(1)}_5$} in the {$q$}-{P}ainlev\'e
  {V} equation.
\newblock {\em Funkcial. Ekvac.}, 46(1):173--186, 2003.

\end{thebibliography}

\end{document}